\documentclass[12pt, draftcls, onecolumn,letterpaper]{IEEEtran}

\usepackage{amsmath}
\usepackage{amsthm}
\usepackage{graphicx,cite}
\usepackage{amssymb}
\usepackage{array}
\usepackage{fixltx2e}
\usepackage{amsfonts}
\usepackage{lineno}
\usepackage{newalg}
\usepackage{epstopdf}
\usepackage{setspace}
\usepackage{subfigure}
\usepackage{color}

\theoremstyle{plain}

\newtheorem{prop}{Proposition}
\newtheorem{corollary}{Corollary}
\theoremstyle{definition}
\newtheorem{definition}{Definition}
\theoremstyle{remark}
\newtheorem{remark}{Remark}
\newtheorem{example}{Example}

\newcommand{\R}{\mathbb{R}}

\newcommand{\cM}{{\cal M}}

\newcommand{\off}[1]{}

\long\def\symbolfootnote[#1]#2{\begingroup%
\def\thefootnote{\fnsymbol{footnote}}\footnotetext[#1]{#2}\endgroup}

\makeatletter
\def\ps@headings{%
\def\@oddhead{\mbox{}\scriptsize\rightmark \hfil \thepage}%
\def\@evenhead{\scriptsize\thepage \hfil \leftmark\mbox{}}%
\def\@oddfoot{}%
\def\@evenfoot{}}
\makeatother
\begin{document}
\title{Distributed Inter-Cell Interference Mitigation Via Joint Scheduling and Power Control Under Noise Rise Constraints}

\author{\IEEEauthorblockN{Erez Biton\IEEEauthorrefmark{0},
Asaf Cohen\IEEEauthorrefmark{0},
Guy Reina\IEEEauthorrefmark{0} and Omer Gurewitz\IEEEauthorrefmark{0}  \\
 }
 \IEEEauthorblockA{\IEEEauthorrefmark{0}
Department of Communication Systems Engineering,
Ben Gurion University, \\ (Email: berez, coasaf, guy.reina, gurewitz@bgu.ac.il) }}

\maketitle

\begin{abstract}
Consider the problem of joint uplink scheduling and power allocation. Being inherent to almost any wireless system, this resource allocation problem has received extensive attention. Yet, most common techniques either adopt classical power control, in which mobile stations are received with the same Signal-to-Interference-plus-Noise Ratio, or use centralized schemes, in which base stations coordinate their allocations.

In this work, we suggest a novel scheduling approach in which each base station, besides allocating the time and frequency according to given constraints, also manages its uplink power budget such that the aggregate interference, ``Noise Rise", caused by its subscribers at the neighboring cells is bounded. Our suggested scheme is distributed, requiring neither coordination nor message exchange.

We rigorously define the allocation problem under noise rise constraints, give the optimal solution and derive an efficient iterative algorithm to achieve it. We then discuss a relaxed problem, where the noise rise is constrained separately for each sub-channel or resource unit. While sub-optimal, this view renders the scheduling and power allocation problems separate, yielding an even simpler and more efficient solution, while the essence of the scheme  is kept. Via extensive simulations, we show that the suggested approach increases overall performance dramatically, with the same level of fairness and power consumption.
\end{abstract}

\IEEEpeerreviewmaketitle
\section{Introduction}\label{sec:intro}
The desire to provide integrated broadband services while maintaining \textit{Quality of Service} (\emph{QoS}) guarantees growing interest in scheduled access techniques used in multiple-access protocols for future broadband radio systems. Such schedule-based techniques are utilized to ensure that a transmission, whenever made, is not hindered by any other transmission and is therefore successful. Accordingly, \emph{Orthogonal Frequency-Division Multiple Access} (\emph{OFDMA}) has been widely adopted as the core technology for various broadband wireless data systems, including the next generation cellular systems, 3GPP \emph{Long Term Evolution} (\emph{LTE}), \cite{sesia2011lte}, and IEEE 802.16e/m (WiMAX), \cite{IEEE_802.16e,IEEE_802.16m}. In these systems, the \textit{base station} (\textit{BS}) allocates (schedules) distinct frequency-time chunks among the active \textit{mobile stations} (\textit{MS}) within its cell, both for their downstream (BS to MS) and for their uplink (MS to BS) traffic. In addition to the frequency-time allocation, the BS also determines the uplink transmission power of the preselected (scheduled) MS, a.k.a.\ \emph{uplink power-control}.

Common power-control approaches to the uplink resource allocation problem are either to assign transmission power to the MSs such that all are received at the BS with the same \emph{Signal to Interference-plus-Noise Ratio} (\emph{SINR}) \cite{Zan92,Yates95}, or to allow MSs to transmit at their maximal available power \cite{RBUL09}. Both of these techniques optimize the MS/cell throughput (intra-cell throughput), neglecting the interference injected to neighboring cells (inter-cell interference).

Note that in the more common modes of operations today, duplexing is achieved via either time or frequency domain division. Thus, the uplink and downlink transmissions are separated, and when considering, e.g., uplink inter-cell interference, one has to consider only the interference caused at a given BS by transmitting users in the surrounding cells, and not the transmissions of other BSs close by.

Since OFDMA systems are sensitive to inter-cell interference, the interference from neighboring cells can dramatically decrease the SINR received at the BS, hence reduce the MS throughput. Moreover, without knowing in advance the interference a BS is expected to experience in a transmission, an MS is unable to fine-tune its modulation and coding scheme to the expected SINR at the receiving BS. Accordingly, power control plays a decisive role in providing the desired SINR, not only by controlling the MS received signal strength at its intended BS, but also by controlling the interference caused to neighboring cells. This double role  is challenging, as on the one hand as far as intra-cell throughput is concerned, an MS in the proximity of the BS is expected to have high quality link, hence high throughput, even when transmitting in low power, while a distant MS needs to transmit at much higher power to attain the same throughput, and on the other hand as far as inter-cell interference is concerned, MSs near a BS can transmit at high power since they are not in the proximity of other cells, while distant MSs which can be in the proximity of other cells should not transmit at high power as they can interfere with other (neighboring) BSs.

In order to limit the interference to neighboring cells, 3GPP has approved the use of \textit {Fractional Power Control} (\emph{FPC}) \cite{FPC08}. According to this approach, MSs with higher path-loss, which are expected to be far from their BS, should operate at a lower SINR requirement so that they will generate less interference to neighboring cells. Nonetheless, according to this approach, in order to maintain some notion of fairness, most of the resources should be allocated to far-away MSs which will transmit in lower modulation schemes. Similarly, the IEEE 802.16m power control scheme deducts a fraction of the downlink signal-to-interference ratio (SIR) from the transmission power. By that, it reduces the interference caused by cell-edge MSs. The 3GPP LTE standard suggests a different approach for combating inter-cell interference termed \textit{Inter-cell Interference Coordination} (\emph{ICIC}) (e.g., \cite{ICIC09}). ICIC provides tools for dynamic inter-cell-interference coordination of the scheduling in neighboring cells such that cell-edge MSs in different cells are preferably scheduled in complementary parts of the spectrum. However, ICIC requires coordination between neighboring cells, both in terms of exchanging information regarding subscribers at one cell and their interference level on other neighboring cells, as well as coordination in the resource allocation, which further complicates the scheduling process.

In this work, we introduce a different approach, which controls the inter-cell interference, yet does not require any cross deployment communication or coordination. In our approach, the aggregate uplink inter-cell interference that all MSs in a cell are allowed to induce is bounded. This limited egress interference budget, termed Noise Rise, is treated as an additional limited resource which is allocated to MSs by the BS in conjunction with the ordinary resources (time and frequency), according to some fairness criterion and channel condition. We show that controlling the interference generated by each cell also controls the average interference level sensed by each BS and provides a more predictable uplink SINR, which allows lower interference margins and more efficient rate selection. Hence, it obtains higher capacity and better coverage. In particular, our contributions are as follows.

First, we introduce the Noise Rise concept which bounds the aggregate uplink interference that all MSs in a cell are allowed to interfere with all the surrounding cells. We suggest means for a BS to estimate the normalized interference of each of its MSs, and show that by limiting the egress interference the ingress interference is controllable. We show that by utilizing the Noise Rise concept we can solve the joint scheduling problem between all BSs in the network distributively, by each BS independently from even neighboring BSs.

Second, we formalize the scheduling problem under the noise rise constraint as a convex constrained optimization problem, and provide an efficient iterative algorithm that is proved to solve it optimally. Moreover, we suggest a second setting, in which instead of bounding the average noise rise over all channels, allowing some sub-channels to contribute more noise rise at the expense of further limiting the noise rise on others, we bound the noise rise on each sub-channel to the exact same value. The latter setting allows the decoupling of the scheduling algorithm from the power control and thus facilitates an even simpler algorithm.

Third, we thoroughly evaluate the noise rise concept via an extensive set of simulations, using both an all-inclusive simulator as defined by IMT-Advanced~\cite{ITUM2135} and numerical results for the exact expressions we analyze utilizing the Shannon-capacity based approach. Our numerical results clearly depict that the suggested approach dramatically increases the overall throughput achieved in each cell compared to the traditional approach, while maintaining fairness. The results obtained by the IMT-Advanced simulator include a more realistic setup which takes into account modulation, coding and several other practical aspects, and show that even though MSs in the proximity of the BS (hence can take advantage of transmitting in high power and high modulation rates), lose throughput due to the noise rise constraint, MSs further away from the BS, and in particular those closer to the cell edge, gain dramatically due to the noise rise constraint.

\section{Noise Rise} \label{sec:noiseRise}
In this paper we study the joint uplink scheduling and power control problem for wireless cellular networks. We consider a multi-cell network comprising $\mathcal{B}$ Base Stations (BSs) each serving a set of Mobile Stations (MSs). We denote by $\mathcal{M}_i$ the set of MSs served by BS $i$. The BS deployment is assumed to be fully symmetric, i.e., we assume that the BSs form a two dimensional lattice e.g., hexagonal grid. We assume that the number of MSs ($\mathcal{M}_i$) and their distribution over each cell is \emph{i.i.d}.

Interference is one of the key factors impacting the performance of wireless networks. It can be partitioned into external interference, which is caused by the coexistence of other wireless networks that operate on the same frequency bands and to internal interference, which is caused by other transmissions within the same network. In cellular networks, internal interference can be further partitioned into intra-cell interference, which relates to other transmissions within the same cell, and inter-cell interference, which relates to interference from neighboring cells. Traditional cellular communications such as CDMA networks suffer from Intra-cell interference due to the pseudo orthogonality of the CDMA codes used within a cell. Such intra-cell interference is resolved in OFDMA systems due to the orthogonality characteristics of the subcarriers in these systems.  Nonetheless, OFDMA technology does not provide any solution to the inter-cell interference. Particularly, the received SINR on one cell is highly dependent on neighboring cell transmissions, i.e., the received SINR on a specific resource unit (here we denote by resource unit the smallest time-frequency resource allocation used for downlink/uplink transmission) highly depends on the scheduled MSs on the same resource unit and their associated transmission powers in the neighborhood cells. Consequently, when each cell performs its own power control and scheduling independent of its neighboring cells (no coordination between the cells), the interference profile seen in uplink transmission becomes highly dynamic, i.e., the interference level at a BS can vary considerably; such variability is very detrimental to the transmission rate and coding scheme selection. Accordingly, in order to maintain low packet drop ratio (PDR) despite this variability, a high interference margin should be accounted for while selecting the transmission rate and coding scheme. That is, transmitting in a more robust rate (lower rate) to maintain reliability at worse cases of the interference pattern.

Clearly, this inter-cell interference effect reduces cells' uplink spectral efficiency, resulting in poor resource utilization and performance. In order to combat co-channel interference (CCI), several studies addressed the problem both for setups in which a single antenna is utilized and for the case in which multiple antennas are utilized. In the context of Single Antenna Interference Cancelation (SAIC) for which most attention was drawn toward GSM networks, the most prominent classes of algorithms suggested are joint demodulation and interference cancelation, e.g. \cite{austin03_SAIC, Nickel2005_SAIC, meyer2006_SAIC}. In the context of Multiple Antenna Interference Cancelation (MAIC) most attention was drawn to intra-cell downlink Multi-User multiple-input multiple-output (MU-MIMO) wireless systems, and was aimed at suppressing the co-channel interference (CCI), by designing multi-user transmit beamforming (or precoding) vectors or matrices which optimize the signal-to-leakage ratio (SLNR), e.g., \cite{tarighat_ICASSP05, zhang2005linear, sadek2007_SLNR, sadek2011_leakage}. Note, however, that the above setting, as well as the resulting optimization problems, are different from the ones we consider in this paper.

In addition to the inter-cell interference effect on cell capacity, inter-cell interference can be also vital for cell coverage, i.e., if noise plus interference is not constrained, due to the maximal power transmission of an MS, certain MSs (typically at the cell edge) will not be able to maintain a reasonable communication capacity. These MSs will be blocked and removed from the coverage area, resulting in reduced cell coverage.

The key concept we wish to consider in this paper is the \emph{uplink noise-rise}.
\begin{definition}
Uplink \textit{noise rise} is defined as the total uplink received noise plus interference power over the background noise power. Formally, let $N_0$ denote the background noise level and $I$ denote the total interference at the BS (receiver). Then, the noise rise $\gamma$ is defined by $\gamma=\frac{N_0+I}{N_0}$.
\end{definition}

In the following subsection we present the constrained noise rise approach, which provides a fully distributed mechanism (without cross deployment coordination and synchronization) that dramatically reduces the variability of the noise rise (due to inter-cell interference) allowing more aggressive rate selection (taking a much smaller noise rise margin).

\subsection{Constrained noise rise approach}
The constrained noise rise approach aims at allowing the scheduler to operate in fixed Noise plus Interference conditions and facilitating a more aggressive rate selection. Specifically, we aim at bounding the egress interference and show that it also bounds the noise rise seen by uplink transmission at all cells.
Formally, let $\mathcal{B}$ denote the set of all BSs in the network, and $\mathcal{M}(k), \; k \in \mathcal{B}$ denote the set of backlogged MSs at BS $k$.
Denote by $I_{in}({k^*})$ the ingress interference level per resource unit at BS $k^*$ receiver,
\begin{equation}\label{eq:interfin}
    I_{in}({k^*})= \sum_{k\in \mathcal{B}\setminus k^*} {\sum_{i(k)\in \mathcal{M}(k)} { L_{i(k),k^*}\cdot p_{i(k)}}}
\end{equation}
where $L_{i(k),k^*}$ denotes the channel gain between MS $i(k)$ of BS $k$ and BS $k^*$ and $p_{i(k)}$ is the transmission power of MS $i(k)$. Note that in order to compute the aggregate interference experienced by BS $k^*$ one has to sum the interference induced to BS $k^*$ by each and every transmission by any MS in any cell in the network other than the MSs in cell $k^*$. Accordingly, the first summation in (\ref{eq:interfin}) is over all BSs besides BS $k^*$ itself, and the second summation is over all scheduled MSs (transmitters) in each such cell.

On the other hand, the egress interference by the MSs of cell $k^*$, which is the aggregate interference induced by each transmitting MS on cell $k^*$ on all BSs other than BS $K^*$ itself, is given by:
\begin{equation} \label{eq:interfeg}
    I_{eg}({k^*})=\sum_{i(k^*)\in \mathcal{M}(k^*)} \sum_{k\in \mathcal{B}\setminus k^*} { L_{i(k^*),k}\cdot p_{i(k^*)}}
\end{equation}
where the first sum is over all MSs in BS $k^*$, and the second sum is the aggregate interference induced by each such transmission on all BSs other than $k^*$ to which the transmitter belongs.

We assume a fully homogeneous deployment which implies that (i) the topology seen by each BS (i.e., the number of neighboring BSs and their location) is \emph{i.i.d} (ii) the backlogged MS distribution in different cells is also \emph{i.i.d}, given that all BSs deploy the same scheduling strategy (e.g., power control and link adaptation mechanisms). Note that both assumptions relate to the distribution and do not require identical spreading of MSs within each cell. We now show that in such a deployment, if the average egress interference caused by the MSs in each cell to the surrounding BSs is the same for all cells, then the average ingress interference experienced by each BS due to its neighboring cells is fixed and equal to the aforementioned average egress interference. That is,

\begin{eqnarray}\label{eq:interfineg}
    {\rm E} \left[ I_{in}({k^*})\right]
    & & = {\rm E} \left[ \sum_{k\in \mathcal{B}\setminus k^*} {\sum_{i(k)\in \mathcal{M}(k)} { L_{i(k),k^*}\cdot p_{i(k)}}} \right] \nonumber \\
    & & = \sum_{k\in \mathcal{B}\setminus k^*}{\rm E} \left[ \sum_{i(k)\in \mathcal{M}(k)} { L_{i(k),k^*}\cdot p_{i(k)}} \right]   \nonumber  \\
     & &  \stackrel{\textcolor{red}{A}}{=}  \sum_{k\in \mathcal{B}\setminus k^*} {\rm E} \left[ \sum_{i(k^*)\in \mathcal{M}(k^*)} { L_{i(k^*),k}\cdot p_{i(k^*)}} \right] \nonumber  \\
      & & =  {\rm E} \left[ \sum_{i(k^*)\in \mathcal{M}(k^*)} \sum_{k\in \mathcal{B}\setminus k^*} { L_{i(k^*),k}\cdot p_{i(k^*)}} \right]  \nonumber  \\
      & & = {\rm E} \left[ I_{eg}({k^*}) \right],
\end{eqnarray}

where \textcolor{red}{A} relies on the homogeneity assumptions above, that is, instead of averaging the interference induced by MS transmissions from cell $k$ on BS $k^*$, we average the interference induced by MS transmissions from cell $k^*$ on BS $k$.

Let $l_{i(k)}$ denote the normalized interference (interference per unit power) that MS $i$ of BS $k$ uplink transmission injects to its neighboring cells, that is
\[
{l_{i(k)}} = \sum\limits_{k' \in \mathcal{B} \setminus k} {L_{i\left( {k} \right),k'}}.
\]
When the BS identity (index) is clear from the context, we will omit it and simply denote it by $l_i$.
The aggregate interference caused by all MSs of cell $k^*$,$I_{eg}({k^*})$ , can be written as
\[
I_{eg}({k^*})=\sum_{i(k^*)\in \mathcal{M}(k^*)} {l_{i(k^*)} \cdot p_{i(k^*)}}.
\]
As previously shown, in order to control the average ingress interference and keep its variability low it is sufficient to bound the egress interference caused by each cell uplink transmission. Accordingly, the noise rise constraint on the scheduler and power control allocation we offer to impose is:
\begin{equation} \label{eq:noiserisebound}
    \sum_{i({k^*})\in \mathcal{M}({k^*})} {l_{i(k^*)}\cdot p_{i(k^*)}} \leq I
\end{equation}
where $I$ is the pre-defined, fixed, noise rise constraint.

We note that (\ref{eq:noiserisebound}) implies that the scheduler should consider noise rise as a resource to be allocated to transmitting MSs in the same manner as resource units. Accordingly, each BS has a noise rise budget $I$ which it can distribute between the scheduled MSs and with respect to each MS's egress interference $l_i$. For example, if the egress interference of a given MS is high, the scheduler could reduce its transmission power and increase its allocated bandwidth. Alternatively, the scheduler could allocate a large portion of its noise rise budget on this MS and co-allocate it with another MS that "consumes" less noise-rise such that the noise rise budget is kept.

In the following we provide the means for a BS to estimate the normalized interference.

\subsection{Normalized interference via downlink SINR}
To comply with (\ref{eq:noiserisebound}), each BS $k \in \mathcal{B}$ should have information on the normalized interference $l_{i(k)}$ of its MSs. Typically, a BS cannot directly measure this coefficient. By coordinating between pilot transmissions of different MSs across the network, each MS's normalized interference can be measured by the surrounding cells and reported to the BS (e.g., via Inter Cell Interference Coordination). This, however, requires coordination among the BSs. Such coordination is supported by the 4G technologies. Alternatively, in the case that the network operates on Time Domain Duplexing (TDD) mode, in which uplink and downlink transmissions are separated in time domain each utilizing the entire spectrum in its turn, a BS may estimate the normalized interference from its MSs' downlink channel state reports without the need for inter-cell communication or coordination.

The downlink SIR (signal-to-interference ratio) measured by an MS is given by
\begin{equation}
    SIR^{DL}_{i\left( {k^*} \right)} = \frac {L_{{k^*},i\left( {{k^*}} \right)}{P^{DL}}}{\sum\limits_{k \in \mathcal{B}\setminus k^*} {L_{k, i\left( {k^*} \right)}{P^{DL}}}}
\end{equation}
where $P^{DL}$ is the BS downlink transmission power.
With the channel reciprocity \cite{Smit04}, i.e., $L_{k,i}=L_{i,k}$, we have
\begin{equation}
    SIR^{DL}_{i\left( {k^*} \right)} = \frac{L_{i\left( {k^*} \right),{k^*}}}{\sum\limits_{k \in \mathcal{B}\setminus k^*} {L_{i\left( {k^*} \right),k}} }
\end{equation}
Therefore,
\begin{equation} \label{eq:li_estimate}
    {l_i} = \sum\limits_{k \in \mathcal{B}\setminus k^*} {L_{i\left( {k^*} \right),k}}  = \frac{L_{i\left( {k^*} \right),{k^*}}}{SIR^{DL}_{i\left( {{k^*}} \right)}}
\end{equation}
Finally, we note that in an interference-limited scenario, when the interference is \emph{considerably greater} than the background noise (which is typically the case), one could use the measured downlink SINR instead of the downlink SIR (required in equation \eqref{eq:li_estimate}), neglecting effect of the background noise on the measurement.

\section{System Model} \label{sec:model}

As shown in Section~\ref{sec:noiseRise}, by limiting the egress interference caused by \emph{each cell's} uplink transmission, independently between cells, we control \emph{each cell's} ingress interference and keep its variability low. Accordingly, we focus on a single cell, and consider the problem of resource allocation for the uplink of an OFDMA cell, where a set $\mathcal{M} = \left\{1,...,M\right\}$ of backlogged MSs transmits to the same BS.

Time is divided into equal length time slots (sub-frames according to the IEEE 802.16m terminology). According to OFDMA, each time slot is divided into a set $\mathcal{N} = \left\{1,...,N\right\}$ of basic (time-frequency) logical allocation units termed \textit{resource units}. We distinguish between logical resource units and physical resource units. Specifically, as widely used in WiMAX, we assume that each logical resource unit is mapped to the same size physical resource unit that has undergone partitioning and permutations. The permutation spreads the (logical) resource units' sub-carriers across the whole frequency band. Alternatively, for the LTE Single Carrier Frequency Division Multiple Access (SC-FDMA) we assume the distributed transmissions, where the users occupy different sets of subcarriers. Accordingly, for each MS $i$ the channel quality of all logical resource units are assumed to be \emph{i.i.d}.
Let us denote by $x_i (t)$ the fraction of the frequency band allocated to MS $i$ at time slot $t$, such that $\sum_{i \in \mathcal{M}} {x_i(t)} \leq 1$. Even though as previously explained the frequency band is divided to a fixed number of channels, throughout the analytical part of this paper, we assume that the frequency band is infinitely divisible, i.e., we allow $x_i(t)$ to take any value between zero and one ($0 \le x_i (t)\le 1$). Note that typically, the number of resource units in a frequency band is large (e.g., in IEEE 802.16m, 48 and 92 resource units for a 10MHz and 20MHz bands, respectively). Obviously, the rounding error is a function of the number of MSs allocated at each time slot, i.e., if only few users are scheduled in a time-slot then the error is expected to be small as only a small subset of the resource units are rounded, and if the number of users scheduled in a time slot is high the error is expected to be high. In the simulation part of the paper we examine our results over both the continuous and the quantized allocation setup, and show that since typically the number of users scheduled in each time slot is low, the rounding errors arising from changing continuous frequency allocations into quantized ones does not affect the results dramatically.

The capacity of a band-limited Gaussian channel (considering other interfering signals as noise \cite[Chapter 15]{cover2006elements}) is $W \log(1+\frac{P}{(N_0+I) W})$, where $W$ is the bandwidth in Hz, $P$ is the received power in Watts and $N_0+I$ is the noise plus interference spectral density in Watts/Hz. Now, denote the total bandwidth available by $B$. Accordingly, the bandwidth allocated to MS $i$ is $W= Bx_i(t)$. Further denote by $p_i(t)$ the power allocated to MS $i$ at time $t$ and by $L_i(t)$ the path gain between MS $i$ and its BS (i.e., $P=p_i L_i(t)$).
Note, the scheduler is assumed to have knowledge of the channel gain, comprising the long term parameters of the link between the BS and the MS, such as path loss and shadowing factor, as well as the short-term time-varying spatial fading.

Thus, the capacity of MS $i$ at time $t$, denoted $r_i(t)$, can be written as
\[
r_i(t)= B x_i(t) \log\left(1+\frac{ p_i(t) L_i(t)}{\left(N_0+I(t)\right) B x_i(t) }\right).
\]
Note that based on Section~\ref{sec:noiseRise}, keeping a fixed noise-rise (i.e., fixed egress interference) by each cell ensures that  the Noise plus Interference  has small variance. Hence, the Noise plus Interference is assumed fixed over time.
 Next, for MS $i$ at time $t$, let us denote by $e_i(t)$ the \emph{normalized} received SINR, that is, $e_i(t) = \frac{L_i(t)}{(N_0+I) B}$. Hence,
\begin{equation} \label{eq:rate}
r_i(t) = B x_i (t) \log \left( 1+ \frac {p_i (t) e_i(t) } {x_i (t)} \right).
\end{equation}
Note that a similar formulation was also used in [6, Section III], when the \emph{single-cell} resource allocation problem was discussed (disregarding inter-cell interference).

Typically, $p_i(t)$ is constrained by the maximum power a user can transmit, i.e., $p_i (t)< P_i$. Nonetheless, throughout the theoretical part of this paper we will assume an interference limited scenario, that is, the maximal power a user can transmit with is higher than the maximal power limit resulting from the noise-rise constraint, even if it was the only transmitter, i.e., $P_i \geq \frac{I}{l_i(t)}$, where, $l_i(t)$ is the normalized interference of MS $i$ at time slot $t$. In other words, the desire not to inflict high interference on the neighboring cells is the actual power limit. We assume that $l_i(t)$ is estimated by the BS according to (\ref{eq:li_estimate}), i.e., $l_i(t)=\frac{L_i(t)}{SIR^{DL}(t)}$ where the downlink Signal-to-Interference ($SIR^{DL}(t)$) is available to the BS (scheduler) via the \textit{Channel Quality Indicator} (\textit{CQI}).
In the simulation part of the paper we will also examine the scenario in which the max-power constraint can be lower than the one which is due to the noise-rise constraint.

Throughout this paper we only consider the uplink allocation of resources (fraction of the frequency band and power) to the backlogged MSs. We assume that the scheduler objective is to maximize a weighted sum throughput. Accordingly, at the beginning of each time slot, the scheduler seeks to maximize a (time-varying) weighted sum of the MS rates. We adopt the gradient-based scheduling framework \cite{KeMa98,AgSu02,Sto05,KuWh02,AgBeLASu01}. Specifically, the scheduler solves the following optimization

 \begin{equation} \label{eq:grbasedsched3}
\mathop {\max } \sum\limits_{i \in \mathcal{M}} {{\omega _{i}(t)}{r_{i}(t)}}
\end{equation} 	
where $\omega_{i}(t) \geq 0$ is a time-varying QoS weight assigned to the $i$-'th MS at time $t$.

We concentrate on weights that depend on the average throughput attained by each MS up to the $t$-th slot, and capture some fairness notation. For example, $\omega_{i}(t) = \frac{1}{T_i(t)}$, where $T_i(t)$ is the average throughput of MS $i$ at time $t$, which captures proportional fairness \cite{KeMa98,KuWh02,AgBeLASu01}.

Note that (\ref{eq:grbasedsched3}) must be re-solved at each scheduling instant (e.g., each sub-frame) due to of changes in both the resource unit state and the weights. Consequently, for the ease of presentation, in the following we omit the time index $t$.

\section{Optimal Joint uplink scheduling and power control} \label{sec:optimization}
In this section, we consider the optimal solution to the general problem of joint scheduling and power control. We formalize the optimization problem, characterize the optimal solution and give an efficient iterative algorithm to achieve it.

The problem at hand is as follows. At the beginning of each time slot, the BS schedules a subset of backlogged MSs to available resource units and assigns transmission power to each scheduled MS. The BS aims at maximizing the achievable rate while providing MSs with a fair share of resources according to a predefined fairness metric and maintaining a bounded interference with neighboring cells.

In conjunction with the scheduler and the power control, a rate adaptation mechanism adjusts the transmission rate according to the allocated power. Clearly, the resulting throughput in a time slot (sub-frame) for a given scheduled MS is derived from the allocated resource units and the allocated power. In other words, the scheduling and power allocation are coupled, in devising the allocated throughput, and should be performed jointly. The power control scheme optimizes the tradeoff between allocated rate and contribution to the overall Noise Rise. Typically, MSs far from the BS are required to transmit at high power in order to maintain a reasonable rate. However, these MSs are closer to neighboring cells, hence contribute more to the noise rise in those cells. Nonetheless, cell edge MSs are required to transmit at low power in order to bound the interference with neighboring cells.

To state the optimization problem formally, at the beginning of each time slot the scheduling and power control scheme selects a feasible resource and power allocation tuple $\left(\boldsymbol{x}, \boldsymbol{p}\right)$ (throughout, we use bold symbols to denote vectors) that complies with the noise rise constraints and maximizes a time-varying weight assigned to each MS, i.e.,
\[
\underset{\boldsymbol{x}, \boldsymbol{p}}{\text{maximize}} \left\{ \sum_{i \in \mathcal{M}} \omega_i \cdot r_i(x_i,p_i)\right\}
\]
where $r_i(x_i,p_i)$ is the rate related to the resource and power allocation and $\omega_i \geq 0$ is the time-varying weight assigned to the $i$-th MS at the beginning of the time slot. These weights are the gradient of an increasing concave utility function of each MS. Taking the rate as $r_i(x_i,p_i)={Bx_i \log \left( 1+ \frac {p_i e_i} {x_i} \right)}$ (equation \eqref{eq:rate}), we can formulate the joint power control and scheduling with noise rise constraint optimization problem: for each time slot, find the channel allocated to each MS, denoted by $\boldsymbol{x}= \{x_{1},x_{2},\ldots ,x_{M}\}$, as well as the power assigned to each MS, denoted by $\boldsymbol{p}= \{p_{1},p_{2},\ldots ,p_{M}\}$, such that the total weighted throughput is maximized. That is,
{
\begin{equation} \label{eq:optimization_problem}
\begin{aligned}
& \underset{\boldsymbol{x} , \boldsymbol{p}}{\text{maximize}}
& & \sum_{i \in \mathcal{M}} B\omega_i x_i \log \left( 1+\frac{p_i e_i}{x_i} \right)  \\
& \text{subject to}
& & x_i,p_i \geq 0, \; \forall i \in \mathcal{M}, \\
&&& \sum_{i \in \mathcal{M}} x_i \leq 1,\\
&&& \sum_{i \in \mathcal{M}} l_i  p_i \leq I.
\end{aligned}
\end{equation}
}

\begin{prop} \label{pr:convexity}
    Optimization problem (\ref{eq:optimization_problem}) is convex with linear constraints.
\end{prop}

\begin{proof}
Consider the negative utility function
\begin{equation}\label{eq:negative_utility}
-\sum_{i \in \mathcal{M}} \omega_i x_i \log \left( 1+\frac{p_i e_i}{x_i} \right)
\end{equation}
We first show that this function is convex. To this end, consider a single summand
\[
f(x_i,p_i) =-\omega_i x_i \log \left( 1+\frac{p_i e_i}{x_i} \right).
\]
The Hessian matrix, restricted to the variables $x_i$ and $p_i$ is given by
\begin{displaymath}
\left( \begin{array}{cc}
\frac{\partial f}{\partial x_i^2} & \frac{\partial f}{\partial x_i p_i} \\
\frac{\partial f}{\partial p_i x_i} & \frac{\partial f}{\partial p_i^2}
\end{array} \right) =
\frac{\omega_{i} e_i^2}{(1+\frac{e_ip_i}{x_i})^2 x_i}
\left( \begin{array}{cc}
\frac{p_i^2}{x_i^2} & -\frac{p_i}{x_i} \\
-\frac{p_i}{x_i} & 1
\end{array} \right)
\end{displaymath}
Note that,
\[
\left( \begin{array}{cc}
\alpha_1  &
\alpha_2
\end{array} \right)\left( \begin{array}{cc}
\frac{p_i^2}{x_i^2} & -\frac{p_i}{x_i} \\
-\frac{p_i}{x_i} & 1
\end{array} \right)\left( \begin{array}{cc}
\alpha_1 \\
\alpha_2
\end{array} \right) = \left(\alpha_1\frac{p_i}{x_i}-\alpha_2\right)^2
\]

Hence, for all $(\alpha_1,\ldots,\alpha_{2M}) \in \R^{2M}$, the Hessian matrix of (\ref{eq:negative_utility}), $H(\boldsymbol{x} , \boldsymbol{p})$, satisfies
\begin{equation*}
(\alpha_1,\ldots,\alpha_{2M}) H(\boldsymbol{x} , \boldsymbol{p}) (\alpha_1,\ldots,\alpha_{2M})^T
 = \sum_{i \in \mathcal{M}} \frac{\omega_{i} e_i^2}{(1+\frac{e_ip_i}{x_i})^2 x_i}\left(\alpha_{2i-1}\frac{p_i}{x_i}-\alpha_{2i}\right)^2 \ge 0
\end{equation*}
\end{proof}

When solving \eqref{eq:optimization_problem} for the optimal $\boldsymbol{x}$ and $\boldsymbol{p}$, for fixed $B$ the values of $\{\omega_i\}_{i \in \mathcal{M}}$ and $\{e_i\}_{i \in \mathcal{M}}$ are fixed non-negative reals. Hence, for ease of notation, from this point on, we omit $B$ and focus on the optimization of $\sum_{i \in \mathcal{M}} \omega_i x_i \log \left( 1+\frac{p_i e_i}{x_i} \right)$ for any $\{\omega_i\}_{i \in \mathcal{M}}$ and $\{e_i\}_{i \in \mathcal{M}}$, subject to the constraints. However, note that $e_i$ does depend on the given $B$ and while computing the cell capacity under the resulting schedule one should multiply by the same $B$.

From Proposition \ref{pr:convexity}, it is clear that the optimal solution can be found numerically using standard optimization techniques. For example, it can be found through a similar method to that used in \cite{RBUL09,RBCDMA10}. Yet, this direct approach might be prohibitively complex and unfeasible for practical implementation in commercial BSs. Accordingly, a simpler solution is called for. To this end, we give an efficient iterative algorithm, which uses the \emph{analytical} solutions to two related sub-problems, to solve the above problem. Moreover, we show that the iterative algorithm converges to the global optimum.

\subsection{An Analytic Solution to the Joint Scheduling and Power Control Problem}\label{sec:optimal}
Herein, we show that the optimal solution can be viewed as two \emph{intertwined water-filling-like} problems, facilitating a highly efficient solution which solves the complete optimization problem by fixing a subset of the variables each time (either powers or bandwidth) and solving the resulting water-filling problem. We show that this iterative procedure is bound to converge, and, moreover, give analytical bounds on the possible values of the slack variables in each of the separate water filling problems, enabling us to converge to their solutions using a \emph{fast binary search}.

The optimization problem we discuss is as follows.
\begin{equation} \label{eq:optimization_problem_no_power}
\begin{aligned}
& \underset{\boldsymbol{x} , \boldsymbol{p}}{\text{maximize}}
& & \sum_{i \in \mathcal{M}} \omega_i x_i \log \left( 1+\frac{p_i e_i}{x_i} \right)  \\
& \text{subject to}
& & x_i,p_i \ge 0, \; \forall i \in \mathcal{M}, \\
&&& \sum_{i \in \mathcal{M}} x_i = 1,\\
&&& \sum_{i \in \mathcal{M}} l_i  p_i = I.
\end{aligned}
\end{equation}

Denote $\max\{x,0\}$ by $[x]^+$.
We first consider the analytical solution to this problem. Proposition~\ref{prop:optimal_analytic} below gives a set of equations satisfied by the optimal bandwidth and power allocations.
\begin{prop}\label{prop:optimal_analytic}
Consider the joint power and bandwidth optimization problem in (\ref{eq:optimization_problem_no_power}). The optimal power and bandwidth allocations, $\{p^*_i\}_{i \in \mathcal{M}}$, and $\{x^*_i\}_{i \in \mathcal{M}}$, respectively, satisfy

\begin{equation}
p^*_i = x^*_i\left[\frac{\omega_i}{\lambda_1 l_i}-\frac{1}{e_i} \right]^+ \qquad ; \qquad  x^*_i = [\tilde{x}_i]^+
\end{equation}

where $\tilde{x}_i$, $\lambda_1$, and $\lambda_2$ are the solution to the following set of equations:
\[
\omega_i \log(1+\frac{p^*_i e_i}{\tilde{x}_i}) - \frac{p^*_i e_i \omega_i}{\tilde{x}_i + p^*_i e_i} + \lambda_2 = 0,
\]
\[
\sum_i l_i x^*_i\left[\frac{\omega_i}{\lambda_1 l_i}-\frac{1}{e_i} \right]^+ = I,
\] and
\[
\sum_i x^*_i = 1.
\]
\end{prop}
\begin{proof}
In the proof of Proposition \ref{pr:convexity} we show the convexity of the optimization problem. The proposition will now follow from a straightforward application of the KKT conditions \cite[Section 5.5.3]{boyd2004convex}. Namely, we write
\begin{multline*}
\lefteqn{ L\left(\boldsymbol{p}, \boldsymbol{x}, \lambda_1, \lambda_2, \{\mu_i\}_{i \in \cM}, \{\tilde{\mu}_i\}_{i \in \cM}\right) } \\
= \sum_{i \in \cM} \omega_{i} x_i \log\left(1+\frac{p_i e_i}{x_i}\right) + \lambda_1\left(\sum_{i \in \cM} l_i p_i - I\right)
+ \lambda_2\left(\sum_{i \in \cM} x_i - 1\right) - \sum_{i \in \cM} \mu_i x_i - \sum_{i \in \cM} \tilde{\mu}_i p_i
\end{multline*}
and the proposition follows by requiring $\nabla L = 0$ and that for all $i \in \cM$ we have$ \tilde{\mu}_i p_i =0$, $\mu_i x_i =0$, $\tilde{\mu}_i \ge 0$ and $\mu_i \ge 0$.
\end{proof}

\begin{remark}
A key result of the Noise Rise concept is that users with a high $l_i$ are \emph{less likely} to receive high power, as they may consume a significant share of the noise rise budget. These users can either be compensated by a larger bandwidth, or, in case they are superior in terms of channel statistics and weights, indeed receive the significant portion of the noise rise budget. Of course, the resulting bandwidth and power allocation is a function of all the parameters in the problem, and must be solved using the optimality equations above or the iterative algorithm we suggest below. Yet, to rectify the dependence of the powers on $\{l_i\}$, fix the bandwidth parameters $\{x_i\}$ and assume $w_i=1$ for all $i$. The resulting equation for $p_i$ is
\[
p_i = x_i\left( \frac{I+\sum_i \frac{x_i l_i}{e_i}}{l_i} - \frac{1}{e_i}\right).
\]
The first summand in the parenthesis can be viewed as the \emph{water level}. It is thus clear that in this case, a larger $l_i$ results in a lower water level, hence a lower $p_i$. However, note that the actual result depends on $x_i$, and this value is part of the optimization problem as well.
\end{remark}

While Proposition \ref{prop:optimal_analytic} gives necessary and sufficient conditions for power allocations $\boldsymbol{p}$ and bandwidth allocations $\boldsymbol{x}$ to be optimal, its direct computation is cumbersome, as the equations for both types of variables are intertwined. However, in the next sub-section, we show that the problem in (\ref{eq:optimization_problem_no_power}) can be solved optimally by a highly efficient iterative algorithm, which, unlike standard iterative optimization procedures, does not jointly optimize over all variables, but rather utilizes the fact that when separating the power variables from the bandwidth ones, each optimization problem has a relatively easy water-filling-like analytical solution.
\subsection{An Iterative Algorithm}
The important observation is as follows. Fixing the bandwidth variables $\{x_i\}_{i \in \cM}$, the resulting optimization problem is
\begin{equation} \label{eq:optimization_problem_only_powers}
\begin{aligned}
& \underset{\boldsymbol{p}}{\text{maximize}}
& &\sum_{i \in \mathcal{M}} \omega_i x_i \log \left( 1+\frac{p_i e_i}{x_i} \right)  \\
& \text{subject to:}
& &  p_i  \ge 0, \; \forall i \in \mathcal{M} \\
&&& \sum_{i \in \mathcal{M}} l_i  p_i = I
\end{aligned}
\end{equation}

The solution to this problem is the well-known \emph{water filling}, e.g., \cite[Example 5.2]{boyd2004convex}. Hence, it is easily solvable (note that the weights $\omega_{i} x_i$ and the noise rise constraints $l_i$ only serve as scaling factors, and do not change the essence of the problem in the \emph{separated case}). Fixing the power variables $\{p_i\}_{i \in \cM}$, on the other hand, results in a relatively similar optimization problem, which although involving an implicit equation for each $x_i$, is also straightforward to solve. The iterative algorithm will then alternate between the two solutions, fixing one set of variables based on the results of the previous iteration. A pseudo code of the algorithm follows.

\begin{center}
\begin{minipage}{1pt}
\begin{algorithm}{Iterative-Water-Filling}{\boldsymbol{e},\boldsymbol{l},\boldsymbol{w}, I}
\boldsymbol{x} \= \boldsymbol{x^0}: \textrm{ such that }  x_i^0>0 \, \forall i, \sum_i{x_i^0}=1 \\
\begin{REPEAT}\\
    \lambda_1\= \CALL{Solve}\left(\sum_i l_i x_i\left[\frac{\omega_i}{\lambda_1 l_i}-\frac{1}{e_i} \right]^+ = I\right)\\
    p_i \= x_i\left[\frac{\omega_i}{\lambda_1 l_i}-\frac{1}{e_i} \right]^+ \, \forall i \\
    \lambda_2 , \boldsymbol{x} \= \CALL{Solve}\\
    \left(\sum_i x_i =1,
     \omega_i \log\left(1+\frac{p_i e_i}{x_i}\right) - \frac{p_i e_i\omega_i}{x_i+p_i e_i}+\lambda_2=0\right)
\end{REPEAT} \CALL{Converge} \\
\RETURN {\{\boldsymbol{x},\boldsymbol{p}\}}
\end{algorithm}
\end{minipage}		
\end{center}

%
%
When evaluating Algorithm {\scshape Iterative-Water-Filling}, the two key aspects are complexity and convergence. First, consider the number of operations in each iteration. The first step, as mentioned, is a basic water-filling procedure. The value of $\lambda_1$ can be calculated by first sorting the MSs according to their value of $\frac{l_i}{\omega_{i} e_i}$, then including MSs in ascending order until the ``water level" $\frac{1}{\lambda_1}$ satisfies the noise rise constraint. This is done in $O(M\log M)$. As for the second step, the solution is more involved, since it cannot be solved explicitly. However, as the following proposition states, the solution is monotonic in $\lambda_2$, with \emph{upper and lower bounds} on the value of the optimal $\lambda_2$, hence can be solved efficiently by a \emph{logarithmic time} binary search.

\begin{prop}\label{prop:solving_for_lambda_2}
For each $i \in \cM$, let $\tilde{x}_i$ be the solutions to
\begin{equation}\label{eq:compute_lambda_2}
\omega_i \log\left(1+\frac{p_i e_i}{x_i}\right) - \frac{p_i e_i \omega_i}{x_i + p_i e_i} + \lambda_2 =0
\end{equation}
Then, for all $i\in \cM$, every $\lambda_2 \leq 0$ and any $\omega_{i}\ge 0, p_i\ge 0$  and $e_i\ge 0$, we have:
\begin{enumerate}
\item $\tilde{x}_i \ge 0$  and $\sum_{i \in \cM} \tilde{x}_i$ is monotonically increasing in $\lambda_2$.
\item The value of $\lambda_2$ in (\ref{eq:compute_lambda_2}) such that $\sum_i \tilde{x}_i =1$ satisfies $\lambda_2^{min} \leq \lambda_2 \leq \lambda_2^{max}$, where

\begin{eqnarray*}
&& \lambda_2^{min} = \min_{i \in \cM} \left\{ \omega_{i} \left( \frac{Mp_i e_i}{1+Mp_i e_i} - \log\left( 1+Mp_i e_i\right) \right) \right\}
\\
&& \lambda_2^{max} = \max_{i \in \cM} \left\{ \omega_{i} \left( \frac{p_i e_i}{1+p_i e_i} - \log\left( 1+p_i e_i\right) \right)\right\}
\end{eqnarray*}

\end{enumerate}
\end{prop}
Thus, when solving the bandwidth iteration in Algorithm {\scshape Iterative-Water-Filling}, simply compute $\lambda_2^{min}$ and $\lambda_2^{max}$, and apply a binary search for the value of $\lambda_2$ such that $|\sum_i x_i -1| \leq \epsilon$. The computational cost is $O(M \log(1/\epsilon))$ assuming (\ref{eq:compute_lambda_2}) is solved for $x_i$ in $O(1)$ for fixed $\lambda_2$.

\begin{proof}
To prove item 1) we proceed as follows.

Set $\frac{p_i e_i}{x_i} = \alpha$ and consider the function $f(\alpha) = \log(1+\alpha) - \frac{\alpha}{1+\alpha}$. It is easy to verify that $f(0)=0$ and that $f(\alpha)$ is non-negative and monotonically increasing for any $\alpha \ge 0$. Thus, for any $\frac{\lambda_2}{\omega_{i}} < 0$, the equation $f(\alpha) = -\frac{\lambda_2}{\omega_{i}}$ will have a unique solution at some $\alpha = \frac{p_i e_i}{x_i} > 0$. Hence, $\tilde{x}_i \ge 0$. Moreover, for fixed $p_i, e_i, \omega_{i}$, the smaller $\lambda_2$ is, the larger is the solution to $f(\alpha) = -\frac{\lambda_2}{\omega_{i}}$, that is, the smaller $\tilde{x}_i$, for all $i$, and hence the smaller is $\sum_i \tilde{x}_i$.

To prove item 2) we note the following.

Since $\sum_i \tilde{x}_i$ is monotonically increasing in $\lambda_2$, the value of $\lambda_2$ such that $\sum_i \tilde{x}_i = 1$ is clearly upper bounded by the value of $\lambda_2$ for which \emph{the ``weakest" MS, the MS which requires the largest $\lambda_2$} in order to achieve $\tilde{x}_i = 1$, indeed gets it (since in this case all other MSs will have $\tilde{x}_{i'} > 1$ and $\sum_i \tilde{x}_i$ will clearly surpass $1$). Thus, $\lambda_2$ is at most $\max_{i \in \cM} \left\{ \lambda_2 \quad \text{s.t.} \quad \tilde{x}_i =1\right\}$, which equals to
\[
\max_{i \in \cM} \left\{\omega_{i} \left( \frac{p_i e_i}{1+p_i e_i} - \log\left( 1+p_i e_i\right) \right)\right\}.
\]
On the other hand, the value of $\lambda_2$ such that $\sum_i \tilde{x}_i = 1$ is clearly lower bounded by the value of $\lambda_2$ for which the ``strongest" MS, the MS which requires the minimal $\lambda_2$ in order to have $\tilde{x}_i = 1/M$, indeed gets it (since in this case all other MSs will have $\tilde{x}_{i'} < 1/M$ and thus $\sum_i \tilde{x}_i$ will be strictly smaller than $1$). As a result, the optimal $\lambda_2$ is at least
$\min_{i \in \cM} \left\{ \lambda_2 \quad \text{s.t.} \quad \tilde{x}_i =1/M\right\}$, which equals to
\[
\min_{i \in \cM} \left\{\omega_{i} \left( \frac{Mp_i e_i}{1+Mp_i e_i} - \log\left( 1+Mp_i e_i\right) \right)\right\}.
\]
\end{proof}

Finally, we mention two important results on the convergence of Algorithm {\scshape Iterative-Water-Filling}. As mentioned, the algorithm iteratively solves two optimization problems, each one involving half of the parameters to be optimized (either powers or bandwidth allocations). This is an alternating optimization procedure. While this procedure can fail for some utility functions (e.g., $x^2-3xy+y^2$ when alternating between the optimization on $y$ and the optimization on $x$), it is important to note that in the specific case of our joint power and bandwidth scheduling with noise-rise constraint, it is bound to converge.
\begin{corollary}\label{cor:convergence}
Assume Algorithm {\scshape Iterative-Water-Filling} is used to solve (\ref{eq:optimization_problem_no_power}). Then the power assignments $\bf p$ and channel allocations $\bf x$ converge to the global optimum of (\ref{eq:optimization_problem_no_power}).
\end{corollary}

\begin{proof}
We first show that the negative utility function
\[
-\sum_{i \in \cM} \omega_{i} x_i \log\left(1+ \frac{p_i e_i}{x_i}\right)
 \]
satisfies an existence and uniqueness constraint \cite[Section 2]{bezdek2002some}. That is, fixing any $2M-1$ variables and optimizing on the remaining variable, the resulting problem has a unique (global) minimizer in the range. This is easily seen from the Hessian matrix calculated in the proof of Proposition~\ref{prop:optimal_analytic}, as the negative utility function is convex in each of the variables.

We now mention that the transform $(\bf p, \bf x)_t = T[(\bf p, \bf x)_{t-1}]$ defined by one iteration of Algorithm {\scshape Iterative-Water-Filling} has no fixed points (for which $T[\bf v] = \bf v$) besides the global optimum of (\ref{eq:optimization_problem_no_power}). This is, again, since the negative utility function is convex. The corollary will now follow by applying \cite[Theorem 2]{bezdek2002some}.
\end{proof}

Note that the result in \cite[Theorem 2]{bezdek2002some} implies that as long as the power and bandwidth allocations $\bf p, \bf x$ are not in the range of the \emph{fixed points} of $T[\bf p, \bf x]$ (an iteration of the algorithm), the negative utility function is strictly decreasing in each iteration. Moreover, the result in Corollary~\ref{cor:convergence} can be made even stronger if the negative utility function is \emph{strictly convex} in the range. Consider, for example, the two MS case. While the function
\[
- \omega_1 x_1 \log\left(1+ \frac{p_1 e_1}{x_1}\right)- \omega_2 x_2 \log\left(1+ \frac{p_2 e_2}{x_2}\right)
\]
is not strictly convex for all $p_1,p_2,x_1,x_2$ (the Hessian matrix is not of full rank and hence not positive definite), \emph{it is strictly convex} under the constraints $x_1+x_2 =1$ and $l_1 p_1 + l_2 p_2 =I$ (for positive $x_i$ and $p_i$). In this case, by \cite[Theorem 3]{bezdek2002some}, the alternating optimization in Algorithm {\scshape Iterative-Water-Filling} \emph{converges $q$-linearly} to the global optimum, from any starting point $(p_1,p_2,x_1,x_2)^0$ in the range. That is, each iteration of the algorithm decreases the distance to the global optimum by a constant multiplicative factor $q \in [0,1)$. In other words,
$
||(\bf p, \bf x)_t - (\bf p^*, \bf x^*)|| \leq q ||(\bf p, \bf x)_{t-1} - (\bf p^*, \bf x^*)||
$.

The algorithm performance, as well as interesting insights on the structure of the utility function we use, are easily visible in the following two-users example.
\begin{example}
Consider the following two-users example. Set $x_1=x=1-x_2$ and $p_1 = p = \frac{I-l_2 p_2}{l_1}$. To fix the constants, set $I=4$, $l_1=4$ and $l_2=1$, so $p_2 = 4 - 4p$ and hence both $x$ and $p$ are in the range $[0,1]$. Of course, this reduces the generality of the problem, yet as it turns out, still results in non-trivial solutions. The optimization problem is thus to find the pair $(x,p)$ which maximize
\[
w_1 x \log\left(1 + \frac{e_1 p}{x}\right) + w_2 (1 - x) \log\left(1 + \frac{e_2(4 - 4p)}{1 - x}\right)
\]
The results for the weights $w_1=1.1$, $w_2 = 9.4$, $e_1=16.25$ and $e_2=0.1$ are given in Figure \ref{fig:two_users}. Note that the constant were chosen to best illustrate a non-tirivial maximum point and the convergence to it. The maximum is achieved at $(x,p)=(0.667419,0.315038)$. This maximum point reflects the balance between the higher interference user 1 creates and its lower weight yet much better channel conditions. This is the reason user 1 receives a significant bandwidth allocation. The three dashed lines represent three runs of the iterative algorithm, from three different starting points. 10 iterations suffice to converge. Moreover, it is clear that a given variable may increase at one iteration, and decrease at the next. That is, the algorithm ``corrects" the bandwidth variables according to the resulting powers, and vice versa.
\end{example}

\begin{figure}
\centering
\includegraphics[scale=0.3]{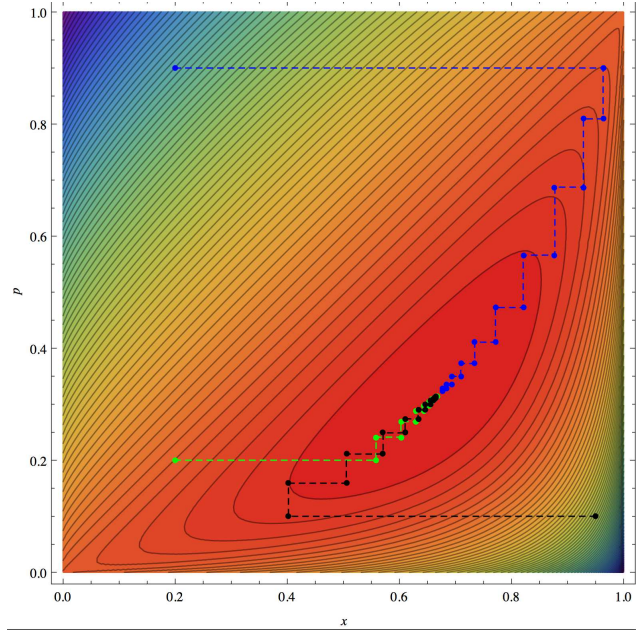}
\caption{A two-users example. A graphical representation of the sum of the wighted capacities as a function of both the bandwidth allocated, $(x,(1-x))$ and the powers distributed, $(p_1=p, \frac{I-l_1 p}{l_2})$. The constants are $I=4$, $l_1=4 l_2 =4$, $w_1=1.1$, $w_2 = 9.4$, $e_1=16.25$ and $e_2=0.1$. On top, the three dashed lines give the results of the iterative algorithm for three different starting points.}
\label{fig:two_users}
\end{figure}

\section{Constrained noise rise density} \label{sec:heuristicA}
In this section, we consider a second noise rise approach, where the \textit{noise rise density} is constrained. That is, instead of noise rise budget over the whole frequency band, the noise rise per resource unit is constrained. We will see that the constraint noise rise approach renders the scheduling and power allocation problems separate, facilitating a very efficient implementation. Moreover, in the separate problems, there is no direct dependence on the Shannon capacity expression and its mathematical properties (e.g., convexity). In fact, the constraint noise rise density approach will be applicable under any \emph{Rate-Adaptation} mechanism, regardless of the actual function connecting the power used to the rate achieved. For example, consider a complicated, real-life scenario where various aspects such as modulations, re-transmissions and error control mechanisms affect the \emph{de-facto} achieved transmission rate. In this case, the Shannon capacity may be far from capturing the actual rates, however, it is possible to devise a function, even if mathematically intractable, that connects the power or bandwidth used to the achieved rate. The approach described herein, will be able to use any such rate-adaptation in the optimization process. Finally, we also mention that this approach is also suitable for the contiguous resource allocation scheme (i.e., without any sub-band partitioning and permutations) and the localized transmission scheme of the LTE SC-FDMA uplink allocation. Both schemes are beneficial for supporting frequency-selective scheduling.

Formally, the constraint on the noise rise density results from normalizing the allocated power by the allocated bandwidth, for each MS $i$. That is, $l_i \cdot \frac {p_i} {x_i} \le I$, for all $i\in \mathcal{M}$.
First, re-writing as $l_i p_i \le I x_i $ and summing both sides over all $i\in \mathcal{M}$, it is clear a resource and power allocation $\left\{ \boldsymbol{x}, \boldsymbol{p} \right\}$ that complies with the noise rise density constraint will \emph{not exceed the noise rise constraint over the whole frequency band}. Namely, the noise rise density approach only adds constraints to the original problem discussed in Section \ref{sec:optimization}.
The scheduling problem with constrained noise rise density is hence as follows.
\begin{equation} \label{eq:density_optimization_problem}
\begin{aligned}
& \underset{\boldsymbol{x} , \boldsymbol{p}}{\text{maximize}}
& & \sum_{i \in \mathcal{M}} \omega_i (t) x_i \log \left( 1+\frac{p_i e_i}{x_i} \right) \\
& \text{subject to:}
& & x_i,p_i \geq 0, \; \forall i \in \mathcal{M} \\
&&& \sum_{i \in \mathcal{M}} x_i \leq 1 \\
&&&  l_i \cdot \frac {p_i} {x_i} \le I, \; \forall i : x_i>0
\end{aligned}
\end{equation}
Now, consider the expression $\omega_i (t) \log \left( 1+\frac{p_i e_i}{x_i} \right)$. Since for each $i$, $\frac{p_i}{x_i} \leq \frac{I}{l_i}$, it is clear that a user which maximized $\omega_i (t) \log \left( 1+\frac{I e_i}{l_i} \right)$ should be allocated the \emph{entire bandwidth}.
The following pseudo code summarizes our solution to the resource allocation problem in this case.
{{
\begin{center}
\begin{minipage}{1pt}
\begin{algorithm}{Resource-Allocation}{\boldsymbol{l},\boldsymbol{w}, I}
\boldsymbol{x} \= \boldsymbol{0}, \boldsymbol{p} \= \boldsymbol{0}\\
\begin{FOR}{i \in \mathcal{M}}
	r_i \= \CALL{Rate-Adaptation}(\frac{I}{l_i})
	\end{FOR}\\
i^* \= \mathop {\arg} \mathop {\max }\limits_{i \in \mathcal{M}} {\left\{ \omega_i \cdot r_i \right\}}, \quad
x_{i^*} \= 1, \quad
p_{i^*} \= \frac {I} {l_{i^*}}\\
\RETURN {\{\boldsymbol{x},\boldsymbol{p}\}}
\end{algorithm}
\end{minipage}		
\end{center}
}}

Algorithm {\scshape Resource-Allocation} is made possible since our assumption on bounded noise rise {\textit{density}} allows the computation of the uplink transmission power density regardless of the schedule. The transmission power density is sufficient for the rate adaptation. Clearly, in practice, a dynamic rate adaptation mechanism would better estimate the achievable rate and facilitate better utilization of the multi-MS diversity resulting in a higher throughput. Algorithm {\scshape Resource-Allocation} is significantly easier to implement, and, as will be made clear in the next section, results in throughput and fairness which do not fall much short than the optimal algorithm.

\section{Simulation results} \label{sec:sim}
In this section, we present the results of thorough simulations carried out to test the performance and behavior of the discussed algorithms, and compare them to existing allocation techniques. First, we include the simulation results of a direct, Shannon-capacity based approach to performance evaluation, where the total throughput measured is simply the capacity of a Gaussian, band limited channel, with the scheduled bandwidth and power, and subject to the noise and interference created by neighboring cells. This simulation clearly illustrates the benefits of the suggested algorithms. Then, to better emulate practical systems, and to assess the algorithm performance under such constraints, we give the results of a thorough simulation, including modulation, coding, packet looses, re-transmissions, and several other practical issues which, while complex to evaluate analytically, can be included in a simulation and give a practical view of the results in this paper.
\subsection{Numerical Results - Shannon Capacity}
In this sub-section, we give the numerical results for the algorithms in Section~\ref{sec:optimal} and Section~\ref{sec:heuristicA}, and compare them to a fixed power scheme. Our main figure of merit is the \emph{total throughput} in the system, summed over all frames and all cells. However, to stress out the benefits of the noise-rise based schemes, the standard deviation of the interference seen at the cells as well as the actual power used are also given.
Moreover, a noise rise density approach \emph{with power constraints} is also simulated.

The deployment included 72 hexagon-shaped cells (base stations), with 722 MSs (MS stations). The MS locations were drawn uniformly at random in space, with the restriction of a minimum of 2 MSs per cell.
Path losses (in dB) were calculated according to the COST-Hata model \cite[Capter 4, equation 4.4.3]{COST-Hata}, specifically:
\[
\begin{split}
  PL =  & \, 46.3 + 33.9 \log_{10}(f) - 13.82 \log_{10}(h_B) - ah_m \\
    &   +(44.9 -6.55\log_{10}(h_b))\log_{10}(d) +c_m
\end{split}
\]

where $f=2000 MHz$ was the transmission frequency, $h_B=50$ the BS antenna effective height, $ah_m$ the MS antenna height correction factor and $c_m$ was taken as $0$ to model a medium size cite. The simulation included 80 frames, with topology and path gains fixed throughout.

The algorithm used to solve the optimal, noise-rise constraint problem (\ref{eq:optimization_problem}) was Algorithm {\scshape Iterative-Water-Filling} (termed Noise Rise in the Figures). For the noise-rise density problem (Section~\ref{sec:heuristicA}), Algorithm {\scshape Resource-Allocation} was used (termed Noise Rise Density). Moreover, a version of Algorithm {\scshape Resource-Allocation} which accounts for a \emph{constant max-power constraint} is also given. That is, we consider the user $i^*$ which maximizes $\omega_i (t) \log \left( 1+\frac{I e_i}{l_i} \right)$ and let $P_{max}$ be the power headroom. If $P_{max} \ge \frac{I}{l_{i^*}}$, we allocate the entire bandwidth to $i^*$. Otherwise, we set $x_i^* = \frac{P_{max}l_{i^*}}{I}$, and allocate the remaining bandwidth to the second best user, and so on. If all users reach their power constraint before all bandwidth is allocated, the algorithm distributes the remaining bandwidth proportionally. Note, however, that this \emph{does not} necessarily result in the optimal solution under these constraints.

In the fixed power scheme we compare with, again, only a single MS per cell is scheduled. However, in this scheme, the allocated power is fixed at the same constant $P$ for all MSs, and only the MS which under such an allocation will have the maximal normalized rate $\omega_{i} \log\left(1+ \frac{Pe_i}{x_i}\right)$ is scheduled. Of course, for a fair comparison, the fixed power $P$ is set such that the mean interference caused by all schemes is the same.

The weights $\omega_i$ were initialized to a constant value, and were updated after each frame according to \cite{KeMa98},
\[
\omega_i(t)= 1/T_i(t),
\]
and
\[
T_i(t)= T_i(t-1) + (1-\beta) \cdot B_i(t-1)
\]
with $\beta =0.9 $ and $B_i(t-1)$ being the number of bits delivered to MS $i$ at sub-frame $t-1$ (0 if MS $i$ was not scheduled at frame $t-1$).

In all four schemes, the actual throughputs after each time slot were calculated according to the SINR measured at the receivers with correspondence to the actual MSs scheduled at this time slot (Shannon capacity for band-limited Gaussian channel). This is to \emph{simulate the exact scenario for which our analytic claims apply}. Note, however, that this is fundamentally different from the all-encompassing simulation described in the next sub-section, where specific modulations and packet loss rates are taken into account.

The results are given in Figures \ref{fig:throughput_std} and \ref{fig:histograms}. In Figure \ref{fig:throughput_std}(a), the total throughput, over all cells and frames is given. The benefit of both noise-rise schemes over the fixed power scheme is clear for all noise-rise values (equivalently, all SNR values). The benefit of the optimal algorithm over the sub-optimal is also clear, as it indeed usually schedules \emph{more than one MS per cell}. Note though, that it is possible to construct topologies in which the benefit is small.

A key advantage in the noise-rise schemes is that the interference seen at the neighboring cells is concentrated around the fixed noise-rise value $I$, allowing them to plan modulation and coding accordingly. Clearly, when the variance of the actual interference observed at the receivers is high, it is harder to choose the appropriate modulation and coding, forcing the senders to either aim at lower modulation and coding schemes or suffer high packet loss. This variance is given in Figure \ref{fig:throughput_std}(b). While the average interference is constrained to the same value in all schemes, the two noise-rise schemes (the optimal and the sub-optimal algorithms) exhibit significantly lower standard deviation. That is, the interference seen under these schemes is centered around the average value, allowing better rate planning by the BSs.

To address the lack of maximum power constraint (headroom) in the optimal (Noise Rise) scheme, the four sub-figures in Figure \ref{fig:histograms} depict the histograms of the actual powers allocated to the MSs. It is clear that a substantial fraction of the MSs is allocated powers \emph{below or around} that allocated by the fixed power scheme, and even the small fraction with high powers \emph{does not require a consequential increase in power}. Thus, although the optimal scheme is not head room constrained, its actual power usage is moderate.

\begin{figure}
\begin{center}
\subfigure[Average throughput]{
\includegraphics[scale=0.5]{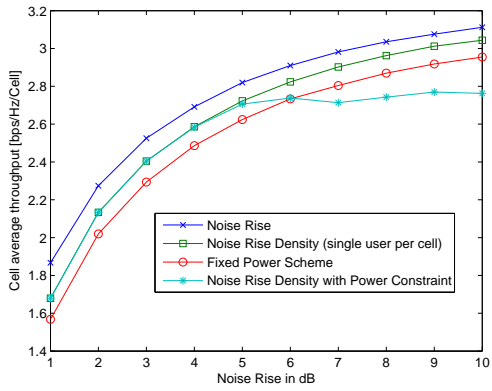}
}
\subfigure[Standard deviation]{
\includegraphics[scale=0.49]{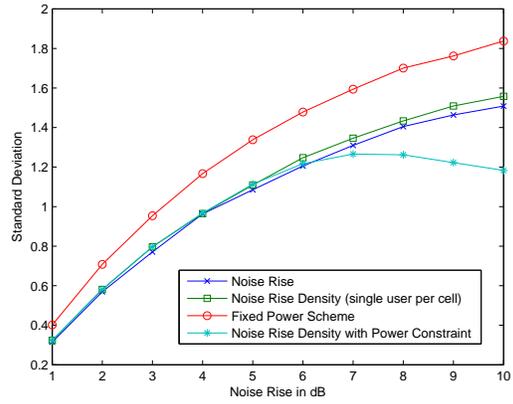}
}
\caption{(a)
Average throughput per cell per frame. (b) Standard deviation of the interference.
\label{fig:throughput_std}}
\end{center}
\end{figure}
\begin{figure}
\centering
\subfigure[2 dB Noise rise]{
   \includegraphics[scale =0.46] {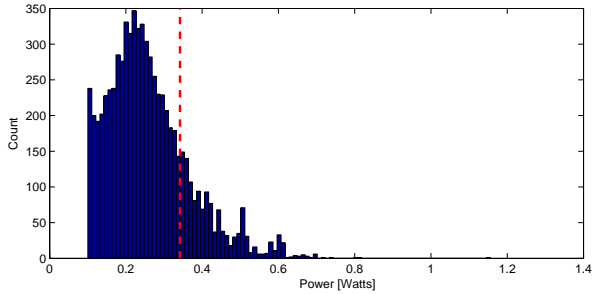}
 }
 \subfigure[5 dB Noise rise]{
   \includegraphics[scale =0.46] {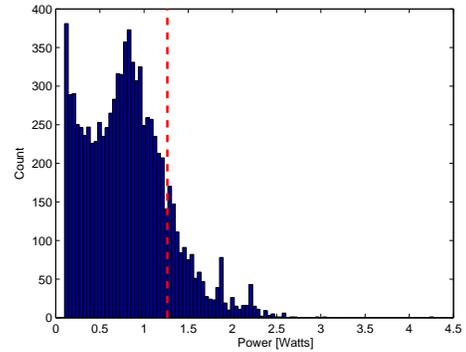}
 }
 \subfigure[7 dB Noise rise]{
   \includegraphics[scale =0.46] {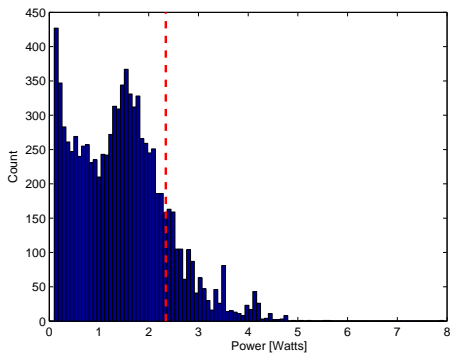}
 }
 \subfigure[10 dB Noise rise]{
   \includegraphics[scale =0.46] {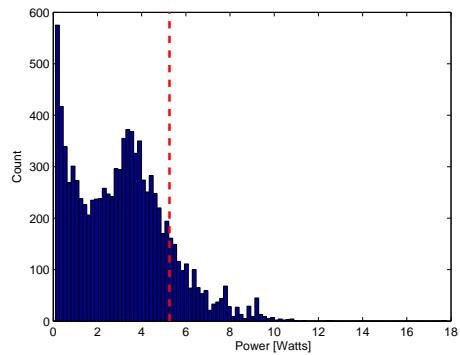}
 }
\caption{Histograms of the allocated powers in the optimal algorithm (only powers above 0.1 are included). The dashed red line represents the power allocated by the fixed power scheme.
\label{fig:histograms}}
\end{figure}

\subsection{An IMT-Advanced Simulation} \label{se:simIMT}

In this section we illustrate the advantages of the noise rise approach via an extensive system simulation, which includes various practical aspects such as finite number of resource units, a more realistic channel model which includes log-normal shadowing model, modulation and coding selection, packet losses, HARQ (Hybrid ARQ) and retransmissions. That is, the system simulation includes more practical aspects that we did not analyze, and did not account for in the numerical results. However these aspects are evident in practical systems and should not be discarded.

To this end, we developed a Matlab simulator that complies with the guidelines for evaluation of radio interface technologies for IMT-Advanced \cite{ITUM2135}.
Specifically, we simulate a two tier hexagonal deployment with 19 sites, each containing three cells (sectors) operating in a frequency reuse 1. Each cell takes its scheduling decisions independently. MSs are located uniformly across the deployment with $10$ MSs in each cell. Statistics are collected from all cells, facilitated by a wrap around geometry such that all cells suffer from interference from two tiers neighboring cells.
Following the IMT-Advanced guidelines, we implemented the WINNER II \cite{WINNERII} stochastic channel modeling.
The modeling approach is based on the geometry of the network layout. The large-scale parameters such as path loss and shadow fading are generated according to the geometric positions of the BS and the MS. Then the statistical channel behavior is defined by distribution functions of delay and angular profiles.
We consider the urban macro-cell scenario, where typically MSs are located outdoors at street level and BSs are fixed clearly above surrounding building heights.  Accordingly, the log-normal shadow fading is assumed with $4dB$ standard deviation for {\em LoS (Line of Sight)} MSs (which are chosen probabilistically as a function of the distance) and $6dB$ for {\em NLoS (Non-LoS)} MSs \cite{ITUM2135}. The assumed inter-site distance is $500m$.
We assume a single user MIMO, with two receive antennas at the BS and one transmit antenna at the MS. The MS and the BS antennas gain are $0$ and $17dBi$, respectively. The assumed noise figure at the MS and at the BS are $7$ and $5dB$, respectively.

We consider the IEEE 802.16m uplink frame structure ({\it Time Domain Duplexing - TDD} mode). The assumed total bandwidth is $10MHz$, occupied by $48$ resource units. The simulation runs sub-frame by sub-frame, performing the scheduling at the beginning of each sub-frame. We ran the simulation over $1000$ sub-frames for $3$ drops.
The {\it exponential-effective SINR Mapping} ({\it EESM}) approach \cite{EESM04} is used to map the system level SINR onto the link level curves to determine the resulting block error rate. Upon block errors, a synchronized retransmission is scheduled with the highest priority after $5$ sub-frames. Upon reception of retransmission the receiver performs chase combining. The maximal number of retransmissions is $4$.

We assume a full-queue model where all MSs have backlogged traffic.
The weights $\omega_i$ were initialized to a constant value, and were updated after each frame according to
\[
\omega_i(t)= 1/T_i(t),
\]
and
\[
T_i(t)= T_i(t-1) + (1-\beta) \cdot B_i(t-1)
\]
with various values of decay factor $\beta$.

The proposed noise rise schemes are compared with two traditional power control mechanisms, namely, (i) maximal transmission power, and (ii) target received SINR. For the first approach, an MS transmits at maximal level regardless of channel conditions, and the link adaptation process assigns the best modulation and coding scheme that maximizes the station throughput given its channel condition. The second SINR based power control approach aims at obtaining a required SINR level at the receiver assuming a fixed noise plus interference scenario. Here, the maximal transmission power and the target SINR for these schemes are set such that the average transmit power is identical to that of the noise rise based schemes.

First, in Figure \ref{fig:NRhistograms}, we compare the ingress noise rise level at the BSs with the four power control schemes. One can see that both the constrained noise rise scheme based on Algorithm {\scshape Iterative-Water-Filling} (termed N.R.) and the constrained noise rise density scheme based on Algorithm {\scshape Resource-Allocation} (termed N.R. density) obtain a relatively narrow histogram around the target noise rise (of 4 db). Alternatively, the maximal transmission power scheme (termed MaxP) and the target received SINR based scheme (termed SINR) result in a much wider histogram (especially the max power scheme). Such wide histograms corresponds to an unpredictably highly variant uplink interference. Note that the max power scheme involves a much larger noise rise, since it allocates more users in a frame, each transmitting at it maximal power over a narrower band.

\begin{figure}
\centering
\subfigure[N.R.]{
   \includegraphics[scale =0.46] {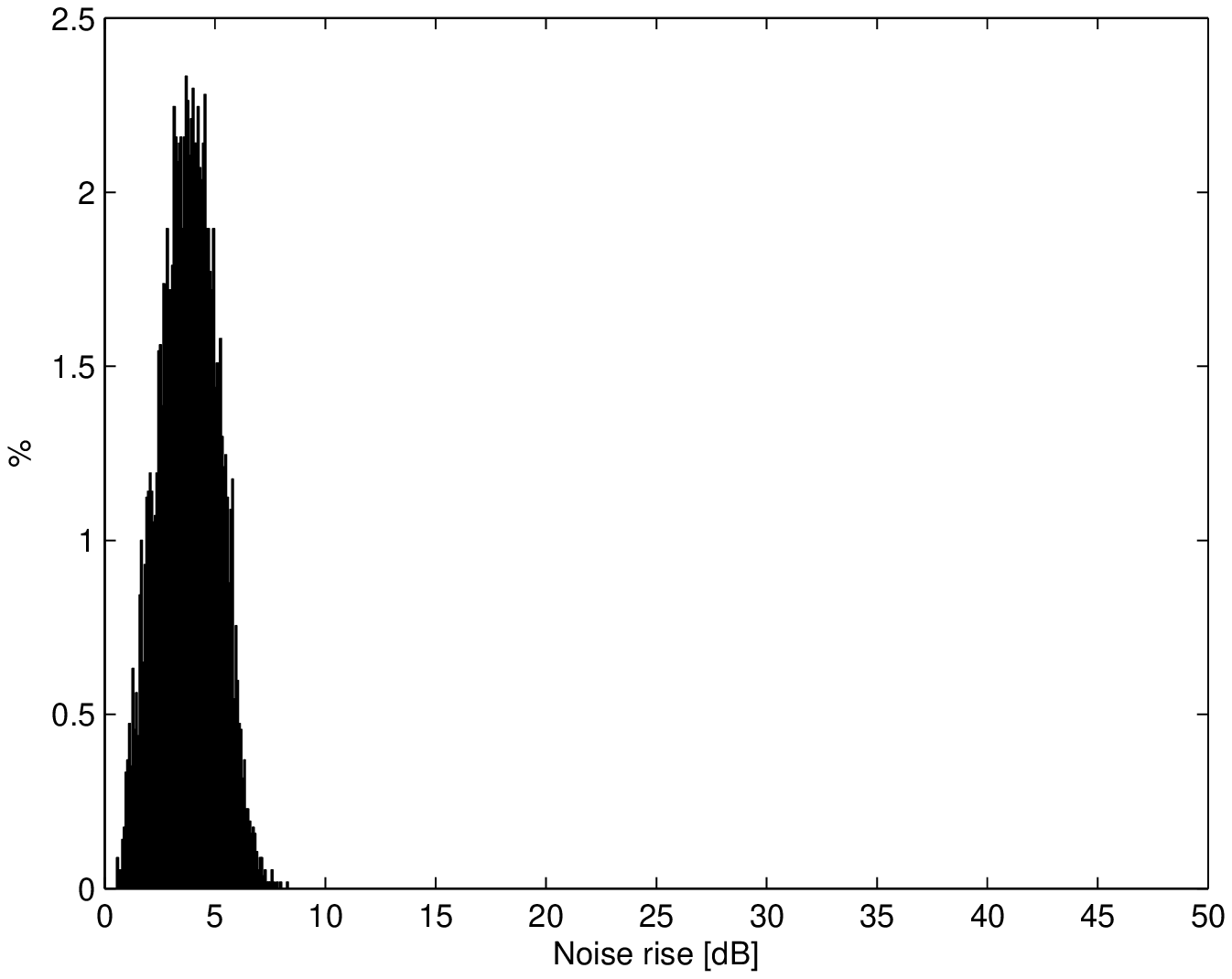}
 }
 \subfigure[N.R. density]{
   \includegraphics[scale =0.46] {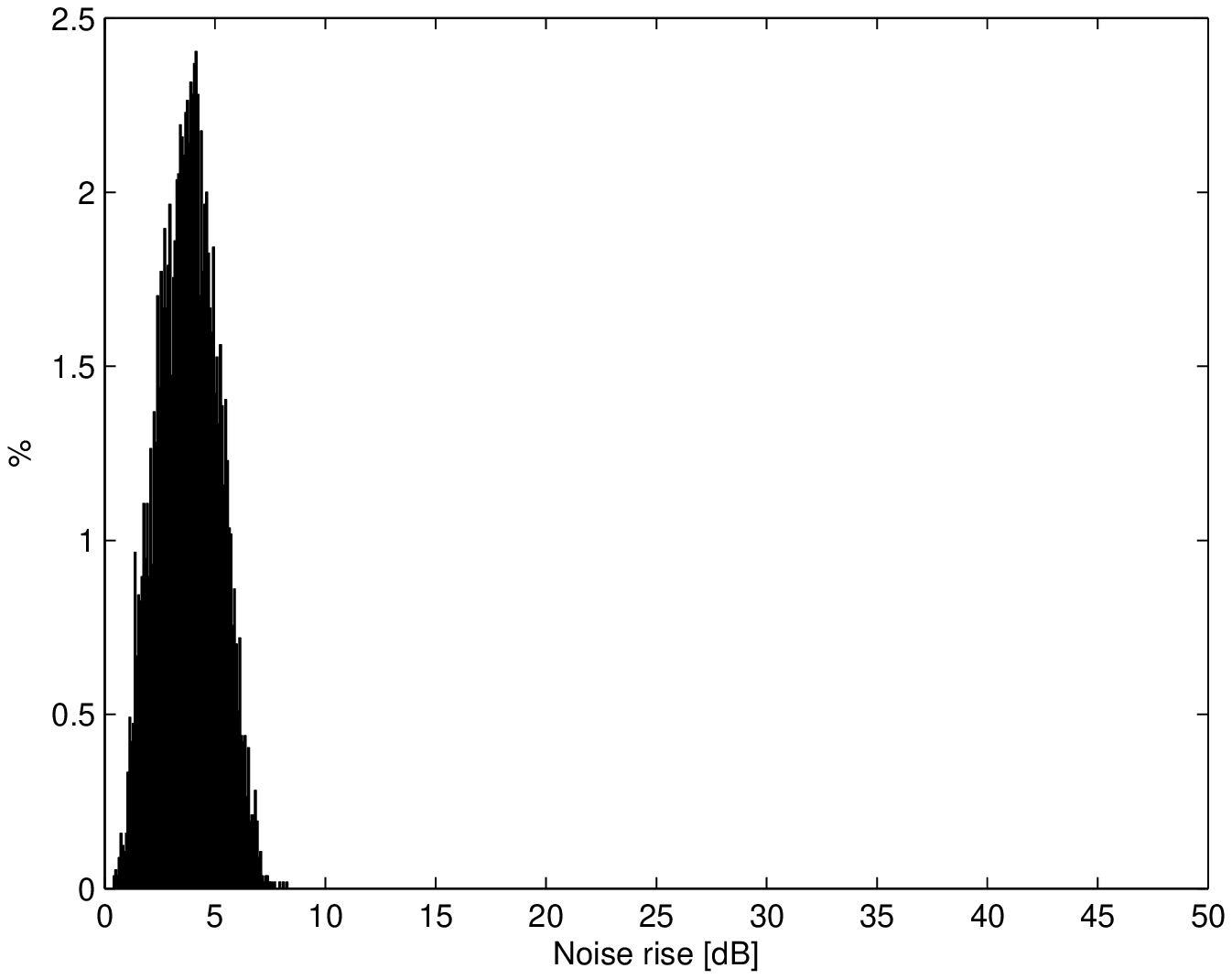}
 }
 \subfigure[MaxP]{
   \includegraphics[scale =0.46] {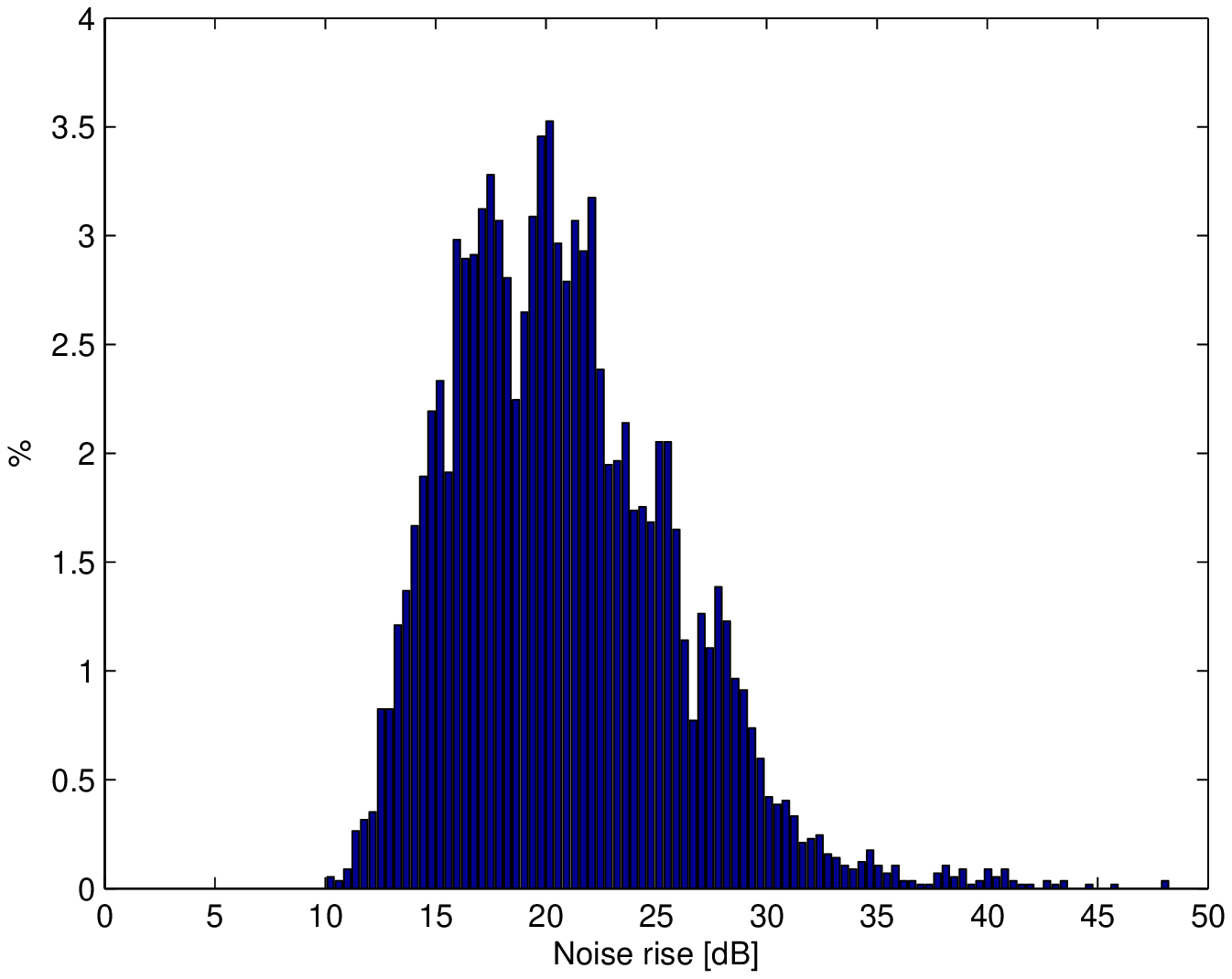}
 }
 \subfigure[SINR]{
   \includegraphics[scale =0.46] {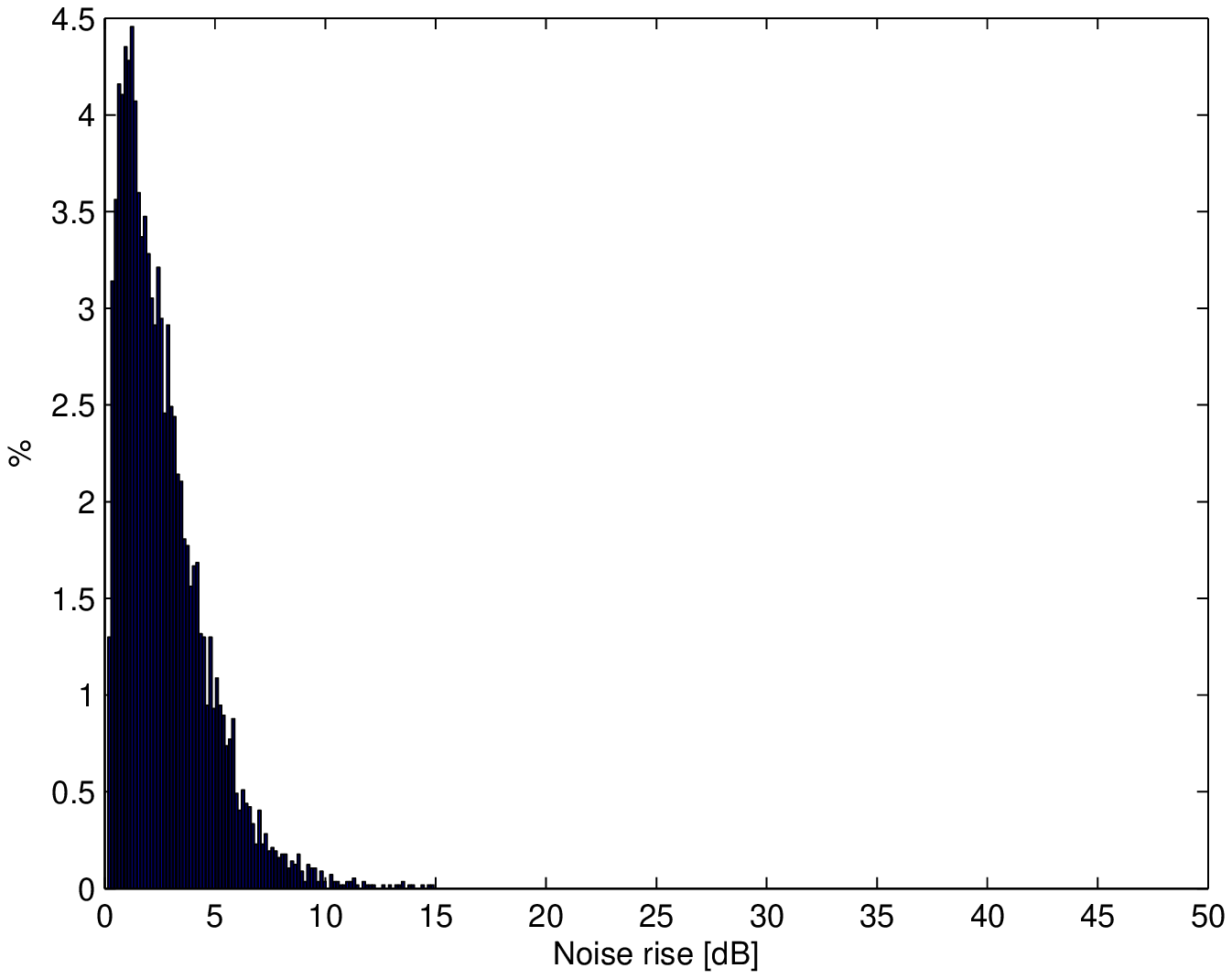}
 }
\caption{Histograms of the ingress Noise Rise in dB with the various power control schemes.
\label{fig:NRhistograms}}
\end{figure}
To compare the performance of the four schemes we adopt two commonly used indicators, namely, the cell and the cell edge MS spectral efficiency (in bits/Sec/Hz). The cell edge MS spectral efficiency is defined as the 5th percentile point of the \textit{CDF} (\textit{cumulative distribution function}) of MS spectral efficiency.

Figure \ref{fig:seBeta} depicts the uplink cell spectral efficiency as well as the cell edge MS spectral efficiency for various values of the proportional fair decay factor $\beta$. Indeed, the greedy approach of transmitting at maximal power obtains the highest throughput. However, it is at the expense of starvation of the cell edge MSs (Figure \ref{fig:seceBeta}). Alternatively, our noise rise schemes provide a better tradeoff between throughput and fairness. For example, the noise rise approach (N.R. in the figure) obtains cell spectral efficiency lower by 40\% from the cell spectral efficiency with the maximal power approach, yet it is 440\% higher than the spectral efficiency with the SINR based approach. Additionally, cell edge MSs with the noise rise approach gain 71\% more throughput than with the maximal power scheme and only 35\% less than the more fair SINR based approach. The schemes' fairness is further illustrated in Figure \ref{fig:se_cdf} that depicts the CDF of the MSs spectral efficiency and a zoom in on cell edge users (Figure \ref{fig:se_cdf_zoom}). Here, it is clear that the SINR based approach provides the best fairness, where all MSs get similar throughput. Clearly, fairness comes at the expense of total system spectral efficiency. On the other hand, the maximal power approach sacrifices about 25\% of the MSs by allocating them unacceptably low throughput (which in practice would result in high blocking probability). Again, the noise rise approach provides a good tradeoff between cell throughput and fairness.
Note that, even though we assumed 1 transmit antenna at the MS and 2 receive antennas at the BS, instead of $2 \times 4$, respectively, the obtained spectral efficiency is only slightly lower than the IMT-Advanced requirements \cite{ITUM2134} (e.g., cell spectral efficiency of \~1 bit/Sec/Hz instead of 1.4 bit/Sec/Hz).

It is interesting to see from Figure \ref{fig:seBeta} that at times the constrained noise rise density algorithm obtains better throughput than the noise rise scheme (e.g., for $\beta=0.6$). This is due to the difference between the simulated IMT advanced EESM channel model and the Gaussian Channel model, which is fundamental to the noise rise approach. Alternatively, the N.R. density approach decouples the scheduling and power control schemes, allowing a link adaptation that does not assume the Gaussian Channel model (see Section \ref{sec:heuristicA}).

\begin{figure}[htbp]
    \centering
    \subfigure[Cell spectral efficiency]{
   \includegraphics[scale =0.47] {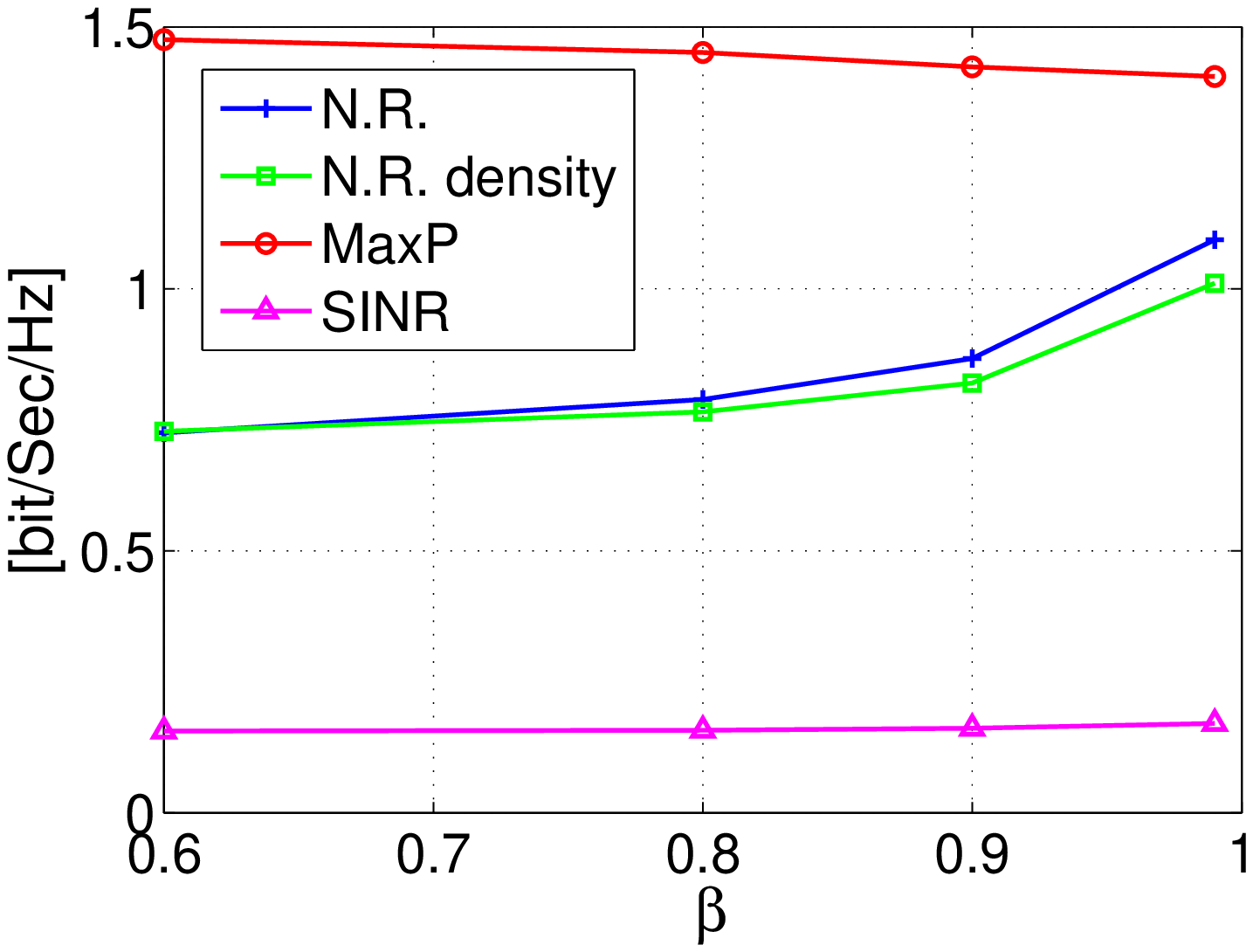}
 }
\subfigure[Cell edge spectral efficiency]{
   \includegraphics[scale =0.47] {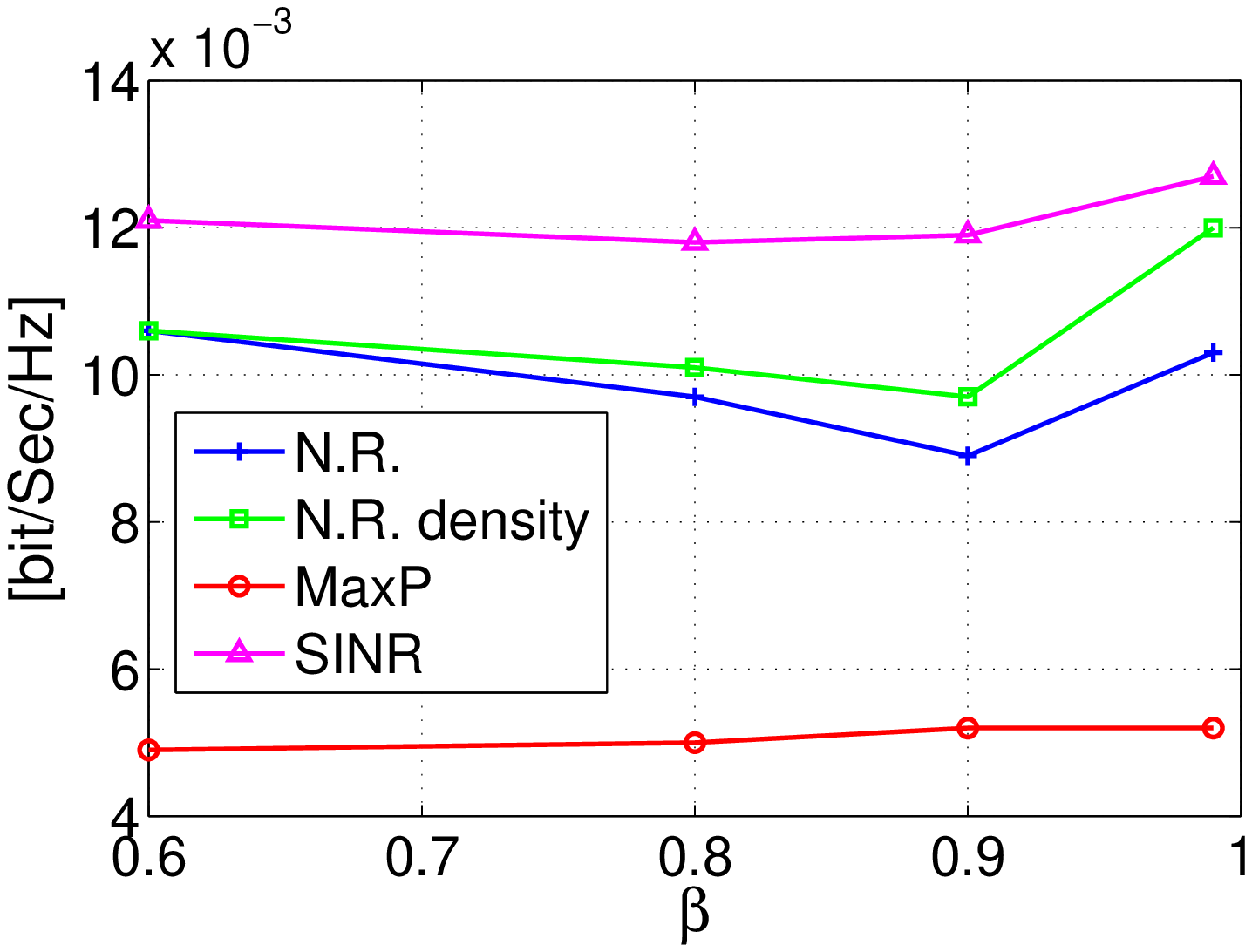}
   \label{fig:seceBeta}
 }
\caption{Uplink cell and cell edge MS spectral efficiency for various values of the proportional fair decay factor $\beta$.
\label{fig:seBeta}}
\end{figure}
\begin{figure}[htbp]
    \centering
    \subfigure[CDF]{
    \includegraphics[scale =0.47] {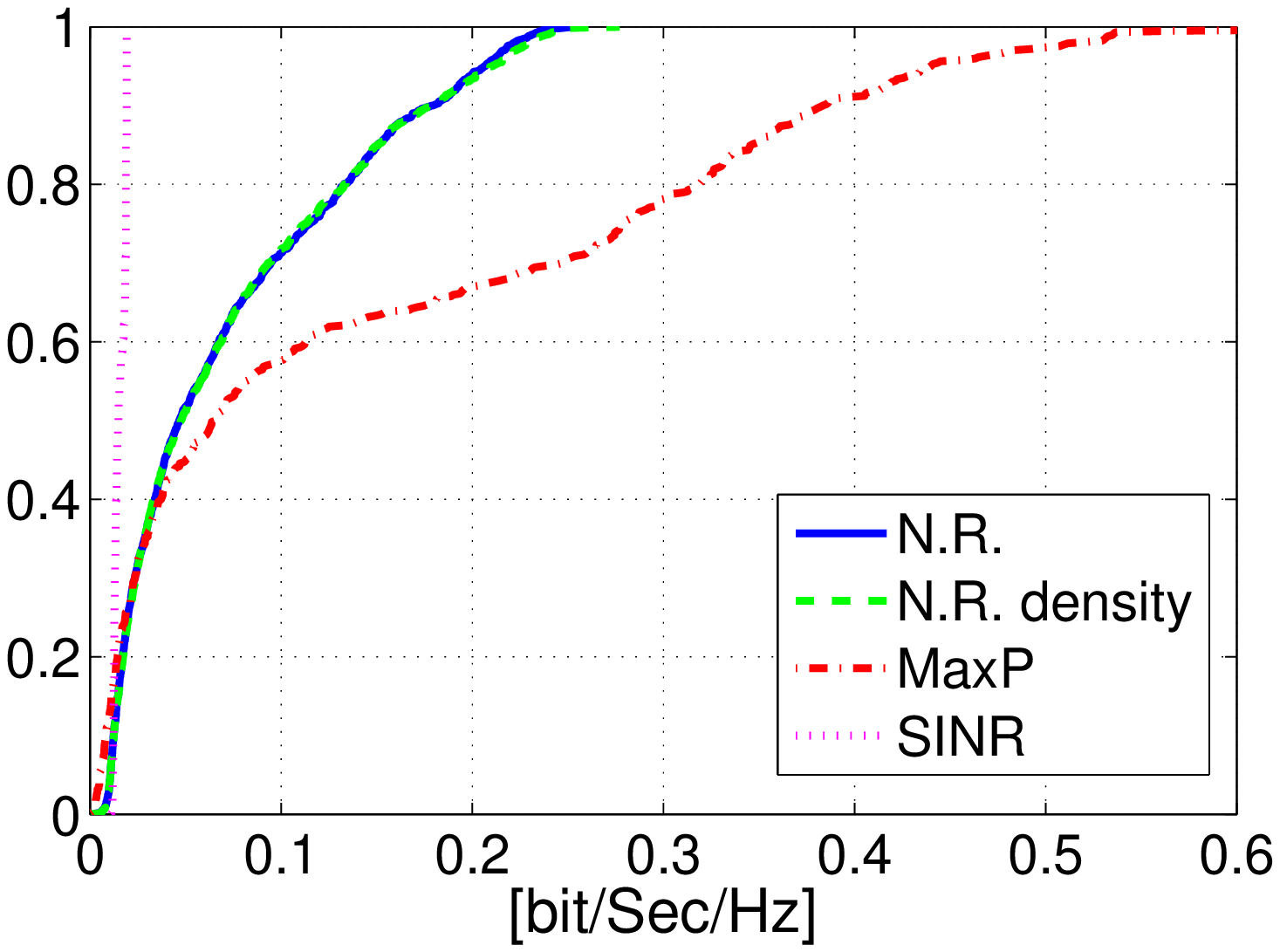}
    \label{fig:se_cdf_full}
 }
 \subfigure[Zoom in]{
    \includegraphics[scale =0.47] {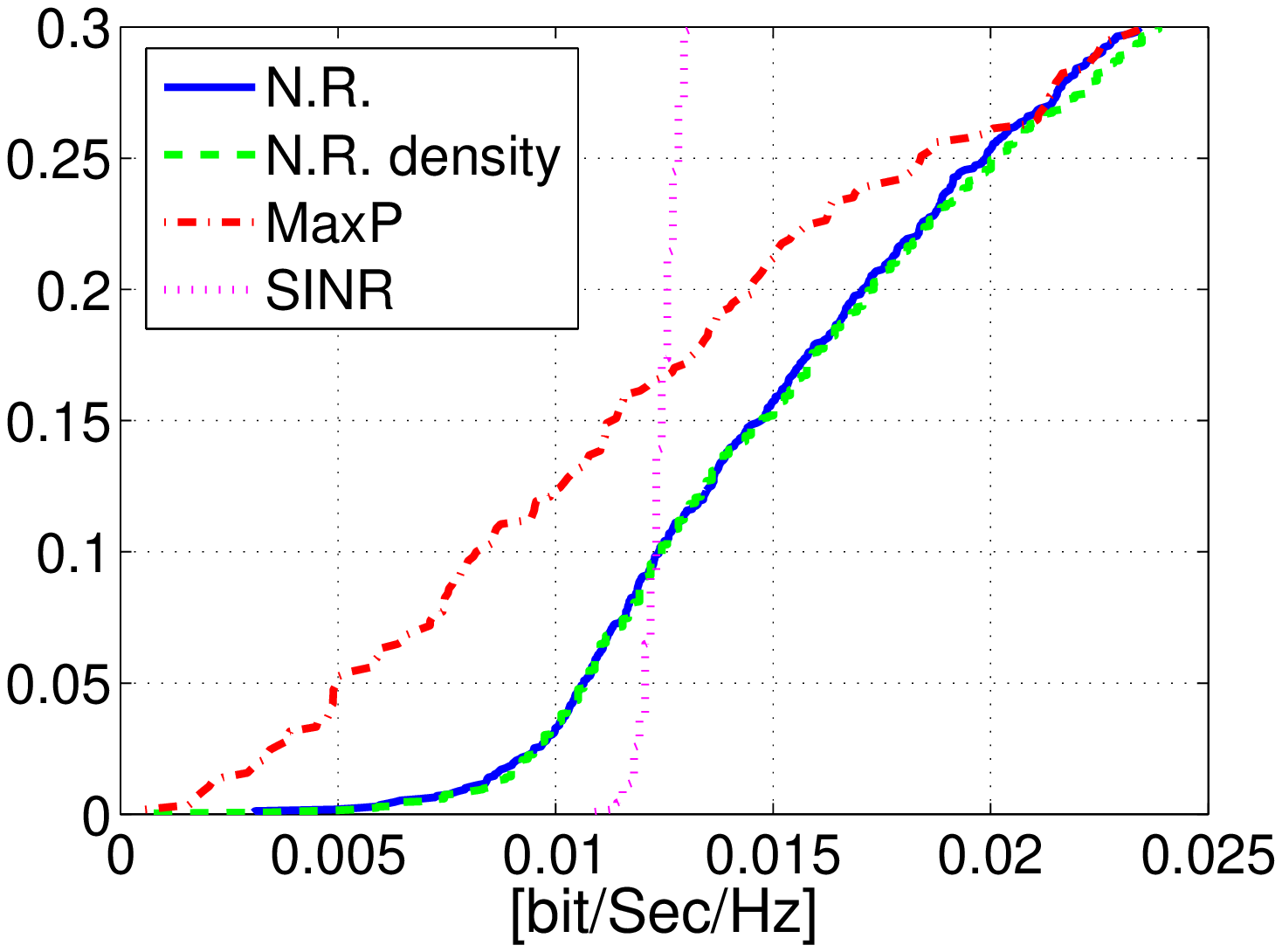}
    \label{fig:se_cdf_zoom}
 }
    \caption{The CDF of the MSs' spectral efficiency, comparing the fairness of the various power control schemes.} \label{fig:se_cdf}
\end{figure}
\begin{figure}
\centering
    \subfigure[Cell spectral efficiency]{
   \includegraphics[scale =0.47] {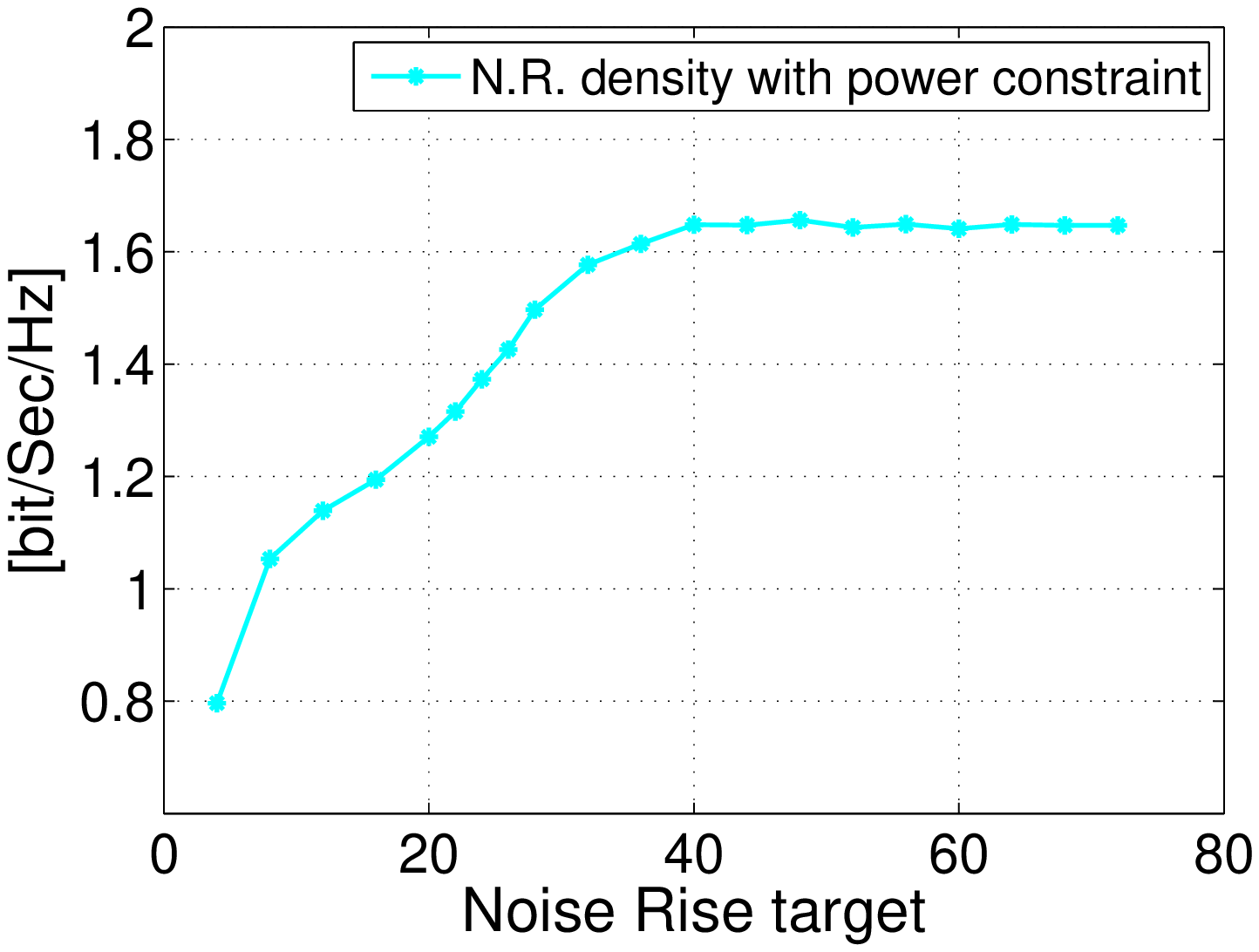}
 }
\subfigure[Cell edge spectral efficiency]{
   \includegraphics[scale =0.47] {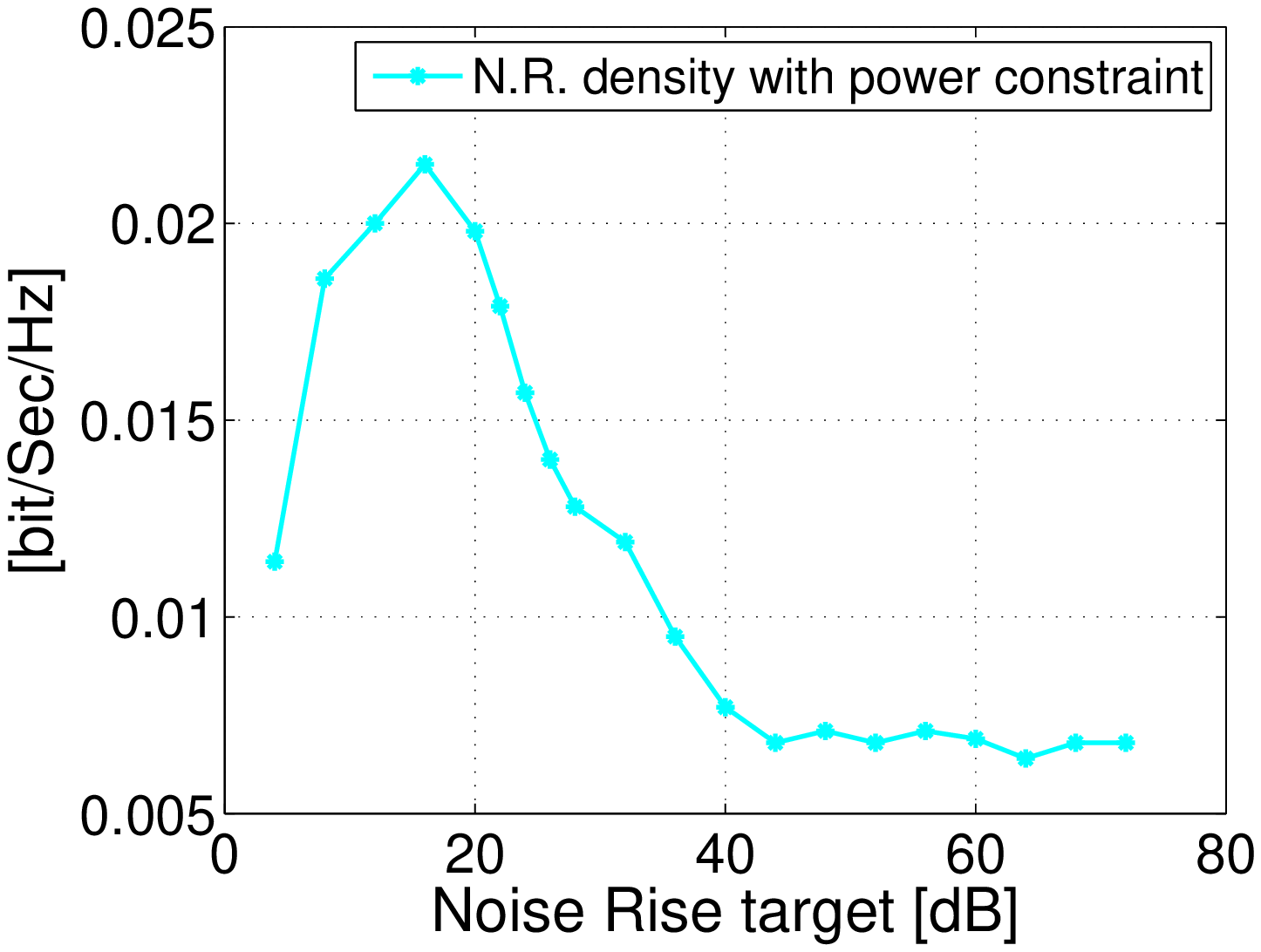}
   \label{fig:seceNR}
 }
\caption{Uplink cell and cell edge MS spectral efficiency for various values of noise rise target.
\label{fig:seNR}}
\end{figure}

Finally, we examine the performance of the noise rise approach with various values of noise rise target. Throughout this paper we assume that the transmission power is constrained by the target noise rise. However, this might not be the case if we increase the noise rise target. Accordingly, to examine high values of noise rise target, we incorporate an additional constraint on the MSs' transmission power. Here, we consider only the noise rise density scheme, where it is easier to incorporate such a constraint.

Following the IMT advanced requirement, we set the MS maximal uplink transmit power to $24 dBm$. Figure \ref{fig:seNR} depicts the cell throughput as well as the cell edge MS throughput for various values of noise rise target. As expected, cell throughput increases with the increase in noise rise. However, as soon as MSs at the cell edge are approaching their maximal power limitations, their throughput decreases. Clearly, at high values of noise rise target, power is constrained by the MS maximal transmission power and the power control operates as the maximal power scheme.

Finally, note that Figure~\ref{fig:seNR} emphasizes the tradeoff in determining the noise rise constraint \emph{a-priori}. On the one-hand, choosing a high noise-rise constraint provides high overall throughput. This, however, comes at the expense of the throughput attained by the cell edge users. On the other hand, selecting a low noise-rise budget will, in general, provide a high cell edge user's throughput, now at the expense of the overall attained throughput.
Of course, choosing a too low noise-rise constraint is detrimental both to the aggregate throughput as well as to the cell edge, since the noise plus interference experienced by the cell edge MSs is dominated by the noise (which remains at fixed level) yet the noise-rise budget is too low for the cell-edge MSs to attain high throughput even at low interference. For example, in the results depicted in Figure~\ref{fig:seNR}(b), edge user throughput is maximized at a noise target of $18dB$.

Accordingly, if the network manager is more concerned about aggregate throughput and less about fairness high noise rise constraint should be preferred. If on the other hand fairness is important lower noise-rise constraint should be selected. Obviously intermediate values which balance between aggregate throughput and cell-edge user throughput are also possible.

\section{Conclusion} \label{sec:conc}
In this paper, we considered a joint scheduling and power allocation problem. Specifically, to mitigate inter-cell interference, we suggested a novel approach, which considers the interference caused by MSs in a BS to neighboring BSs as a resource to be allocated, similar to bandwidth or power. The essence of this approach is as follows. Each MS, based on its channel gains, creates different interference to neighboring cells when transmitting. A BS, while allocating power and bandwidth to its subscribers, does so in a way such that the total aggregated interference its MSs generate does not pass a certain value - its noise rise budget.

We rigorously formulated the problem as a convex optimization problem with linear constraints and suggested an efficient iterative algorithm for its solution, based on known and new water-filling based solutions to its separate problems. We then devised a modified algorithm, which optimizes power and bandwidth allocations assuming the noise rise is constrained for each sub-channel (a portion of the bandwidth allocated to a single MS). This algorithm is highly efficient, attaining higher throughput to cell-edge MSs than other known techniques, while maintaining high overall throughput, approximating the performance of the optimal algorithm. These performance guarantees were also depicted in two extensive simulations, one based on analytical expressions and one based on a practical system built based on the IMT-Advanced specification.

Potential applications of the suggested algorithms include, but are not limited to, uplink scheduling in IEEE 802.16e/m, 3GPP LTE, and more.


\bibliographystyle{IEEE}
\bibliography{NoiseRise}

\end{document}